\documentclass[11pt]{article}

\usepackage{fullpage}

\usepackage[utf8]{inputenc} 
\usepackage[T1]{fontenc}    
\usepackage{url}            
\usepackage{booktabs}       
\usepackage{amsfonts}       
\usepackage{nicefrac}       
\usepackage{microtype}      
\usepackage{xcolor}         
\usepackage{tabularx}  
\newcolumntype{C}{>{\centering\arraybackslash}X}
\makeatletter
\newcommand{\setword}[2]{%
  \phantomsection
  #1\def\@currentlabel{\unexpanded{#1}}\label{#2}%
}
\makeatother
\usepackage{enumerate}

\usepackage{times}
\usepackage{tikz}
\usepackage{mathtools}
\usepackage{amsthm}
\usepackage{amssymb}
\usepackage{bbm}
\usepackage{verbatim}
\usepackage{enumitem}
\usepackage{xspace}
\usepackage{todonotes}
\usepackage{amsfonts,mathrsfs,color}
\usepackage{xcolor}
\usepackage{graphicx}
\usepackage{booktabs}
\usepackage{nicefrac}
\usepackage{mdframed}
\usepackage{contour}
\contourlength{0.1pt}
\contournumber{10}

\usepackage[
  backend=biber,
  style=alphabetic,
  maxbibnames=10,
  maxalphanames=3,
  minalphanames=3,
  natbib=true,
  url=false,
  doi=false]{biblatex}
\bibliography{ref}

\usepackage{multirow}
\usepackage[ruled,linesnumbered,noend]{algorithm2e} 
\SetKwComment{Comment}{/* }{ */}
\SetKwProg{Fn}{function}
\SetAlFnt{\relax}
\SetAlCapFnt{\relax}
\SetAlCapNameFnt{\relax}
\SetAlCapHSkip{0pt}
\usepackage{tikz}
\usepackage{xcolor}
\usepackage{multicol}

\newcommand{\declarecolor}[2]{\definecolor{#1}{RGB}{#2}\expandafter\newcommand\csname #1\endcsname[1]{\textcolor{#1}{##1}}}
\declarecolor{White}{255, 255, 255}
\declarecolor{Black}{0, 0, 0}
\declarecolor{Maroon}{128, 0, 0}
\declarecolor{Coral}{255, 127, 80}
\declarecolor{Red}{182, 21, 21}
\declarecolor{LimeGreen}{50, 205, 50}
\declarecolor{DarkGreen}{0, 80, 0}
\declarecolor{Purple}{146, 42, 158}
\declarecolor{Navy}{0, 0, 128}
\declarecolor{LightBlue}{84, 101, 202}
\usepackage{hyperref}
\hypersetup{
    colorlinks,
    citecolor=DarkGreen,
    linkcolor=LightBlue}
\usepackage[noabbrev, capitalise, nameinlink]{cleveref}
\usepackage{enumerate}

\newtheorem{definition}{Definition}
\newtheorem{theorem}{Theorem}
\newtheorem*{theorem*}{Theorem}
\newtheorem{lemma}{Lemma}

\newtheorem{corollary}{Corollary}
\newtheorem{example}{Example}
\newtheorem{assumption}{Assumption}

\newtheorem{remark}{Remark}

\usepackage{pifont}

\DeclareMathOperator*{\argmin}{argmin}
\DeclareMathOperator*{\argmax}{argmax}

\def\+#1{\mathcal{#1}}
\def\-#1{\mathbb{#1}}

\newcommand{\notshow}[1]{{}}
\newcommand{\AutoAdjust}[3]{{\mathchoice{ \left #1 #2  \right #3}{#1 #2 #3}{#1 #2 #3}{#1 #2 #3}}}
\newcommand{\Xcomment}[1]{{}}

\newcommand{\InParentheses}[1]{\AutoAdjust{(}{#1}{)}}

\newcommand{\InAngles}[1]{\AutoAdjust{\langle}{#1}{\rangle}}
\newcommand{\InNorms}[1]{\AutoAdjust{\|}{#1}{\|}}
\renewcommand{\part}[2]{\frac{\partial #1}{\partial #2}}

\newcommand{\R}{\mathbbm{R}}

\newcommand{\demand}{\bbx^{\mathcal{D}}}
\newcommand{\demandnb}{x^{\mathcal{D}}}

\newcommand{\bbp}{\mathbf{p}}
\newcommand{\bbq}{\mathbf{q}}
\newcommand{\bbx}{\mathbf{x}}

\foreach \x in {A,...,Z}{%
    \expandafter\xdef\csname m\x\endcsname{\noexpand\mathbf{\x}}
}
\foreach \x in {a,...,t,v,w,x,y,z}{%
    \expandafter\xdef\csname m\x\endcsname{\noexpand\mathbf{\x}}
}

\newcommand{\dual}[1]{\tilde{#1}}
\newcommand{\dualdemand}{\dual{\mx}^{\+D}}
\newcommand\numberthis{\addtocounter{equation}{1}\tag{\theequation}}

\title{Fisher Meets Lindahl: A Unified Duality Framework for\\ Market Equilibrium\thanks{Authors are alphabetically ordered.}}
\date{}

\author{
Yixin Tao\thanks{ITCS, Key Laboratory of Interdisciplinary Research of Computation and Economics, Shanghai University of Finance and Economics. Email: \texttt{taoyixin@mail.shufe.edu.cn}} \\
\and
Weiqiang Zheng\thanks{Yale University. Email: \texttt{weiqiang.zheng@yale.edu}}
}

\begin{document}
\maketitle
\begin{abstract}
The Fisher market equilibrium for private goods markets and the Lindahl equilibrium for public goods markets are classic and fundamental solution concepts for market equilibrium. While the Fisher market equilibrium has been well-studied, the theoretical foundations for the Lindahl equilibrium---including characterizations, computation, and dynamics---remain substantially underdeveloped.

In this work, we propose a unified duality framework for market equilibria in private goods and public goods markets. We show that every Lindahl equilibrium of a public goods market corresponds to a Fisher market equilibrium in a \emph{dual} private goods market with dual utilities, and vice versa. The dual utility is based on the indirect utility, and the correspondence between the two equilibria works by exchanging the roles of allocations and prices. This duality framework enables us to transfer insights and results between the two settings. The framework also extends to markets with \emph{chores}.

Using the duality framework, we address the gaps concerning the computation and dynamics for the Lindahl equilibrium and obtain new insights and developments for the Fisher market equilibrium. First, we leverage this duality to analyze welfare properties of Lindahl equilibria. For concave homogeneous utilities, we prove that a Lindahl equilibrium maximizes Nash Social Welfare (NSW). For concave non-homogeneous utilities, we show that a Lindahl equilibrium achieves $(1/e)^{1/e}$ approximation to the optimal NSW, and the approximation ratio is tight. Second, we apply the duality framework to market dynamics, including proportional response dynamics (PRD) and t\^atonnement. We obtain new market dynamics for the Lindahl equilibria from market dynamics in the dual Fisher market, significantly extending existing results for linear utilities. Moreover, the duality framework also introduces new insights into market dynamics. We show that the recently proposed PRD in gross substitutes Fisher markets is a best-response expenditure procedure in the dual Lindahl setting. Using this observation, we extend PRD to markets with \emph{total complements} utilities, the dual class of gross substitutes utilities. Finally, we apply the duality framework to markets with chores. We propose a program for private chores for general convex homogeneous disutilities that avoids the ``poles'' issue, and every KKT point of the program corresponds to a Fisher market equilibrium. We also initiate the study of the Lindahl equilibrium for public chores using duality to the private chores setting.

\end{abstract}
\thispagestyle{empty}
\clearpage
\tableofcontents
\thispagestyle{empty}
\clearpage
\setcounter{page}{1}

\section{Introduction}

The concept of market equilibrium \cite{arrow1954existence, mckenzie1954equilibrium} is fundamental to economic theory, providing a framework for analyzing the interaction of rational agents and efficient allocation of resources. 
From the early work of~\citet{papadeng02}, the study on characterization, computation complexity, algorithm, and dynamics for market equilibrium has become an important topic in algorithmic game theory. In this paper, we study market equilibrium in both private goods markets and public goods markets, and their extensions to chores.

The Fisher market is a classic and extensively studied model for allocating \emph{private} divisible goods among agents with initial budgets~\citep{eisenberg1959consensus, eisenberg1961aggregation}. A Fisher market equilibrium, also known as a competitive equilibrium, is a price for each good with personalized allocations of goods to each agent, such that each agent receives her utility-maximizing bundle of goods under the price and budget constraint, and the total demand of goods equals the supply. A Fisher market equilibrium is Pareto-optimal with many other desirable properties and applications in resource allocation \cite{JainVaz07, jalota2023fisher}. As a fundamental model, the Fisher market has been the subject of extensive investigation, leading to a deep literature on its equilibrium properties~\citep{megiddo2007continuity, chen2012incentive, bei2019earning, goktas2021consumer}, computational algorithms~\citep{codenotti2004efficient, Orlin10, nesterov2018computation, kroer2019computing, garg2022approximating}, and numerous extensions~\citep{branzei2014fisher, Cole:2016:LMG:2940716.2940720, liaostatistical, conitzer2022pacing, gao2023infinite}, just to name a few.

A related but distinct problem is to allocate a fixed budget across divisible \emph{public} goods, which are characterized by non-excludability and non-rivalry. The Lindahl equilibrium, also known as Lindahl tax~\citep{lindahl1958just,foley1970lindahl,  10.1093/restud/rdaf043}, is a solution concept for public goods markets, where agents with initial budgets pay for public goods. A Lindahl equilibrium consists of a fixed allocation of goods and personalized prices (taxes) of goods for each agent, such that the same fixed allocation of public goods maximizes all agents' utilities under their prices and budget constraints, and the total allocation equals the total budget. The Lindahl equilibrium allocation is efficient, always lies in the weak core, and has many applications in collective decision-making~\citep{garg2021markets, brandl2022funding, brandt2023balanced}.

However, our understanding of computation and dynamics for the Lindahl equilibrium for public goods is significantly behind that of the Fisher market equilibrium for private goods. We illustrate this point using the setting where agents have constant elasticity of substitution (CES) utility functions as an example. The CES utility is a class that covers the linear utility and the Leontief utility as special cases (see \Cref{ex:CES and indirect utility} for formal definitions).
\begin{itemize}
    \item Fisher market equilibria are characterized by the Eisenberg-Gale Nash welfare maximization program for the whole class of CES utilities, and, more generally, for homogeneous concave utilities~\citep{eisenberg1959consensus, eisenberg1961aggregation, jain2005market}. But the analog for Lindahl equilibria is known only for special cases of linear, Leontief, and separable nonsatiating utilities~\citep{fain2016core,brandt2023balanced, 10.1093/restud/rdaf043}. \citet{brandt2023balanced} noted that ``\textit{equilibria in Fisher markets are connected to Nash welfare maximization under fairly general assumptions about individual utilities, whereas this connection appears to be more volatile in our public good markets}''.
    \item The proportional response dynamics (PRD)~\citep{WZ2007} is a natural distributed dynamics that converges to Fisher market equilibria for the whole class of CES utilities~\citep{birnbaum2011distributed, cheung2018dynamics}. But for Lindahl equilibria, the convergence of PRD was only recently obtained for the special case of linear utilities~\citep{kroer2025computing}. 
    \item T\^atonnement is a classic price adjustment procedure that computes a Fisher market equilibrium for CES utilities. However, the formulation and convergence properties of t\^atonnement for the Lindahl equilibrium remained unknown and left as an open question in~\citep{kroer2025computing}.
\end{itemize}
On one hand, there is a significant discrepancy of results between the Fisher market equilibrium and the Lindahl equilibrium, as computation and dynamics for the Lindahl equilibrium remained largely unknown.
On the other hand, at least for special cases like the linear utilities, the Lindahl equilibrium admits very similar results to those of the Fisher market equilibrium. The discrepancy and similarity motivate the following question:
\begin{align*}
    \textit{What is the connection between Fisher market equilibria and the Lindahl equilibria?}
\end{align*}

In this paper, we propose a unified duality framework for Fisher market equilibria and Lindahl equilibria. This framework constructs a \emph{dual} private goods market from a public goods market (and vice versa), such that every Fisher market equilibrium in the dual market corresponds to a Lindahl equilibrium in the original market. This duality framework unifies the previously disjoint research directions and allows us to transfer results and insights between the private goods and public goods settings. This new connection and perspective not only enable us to address all the gaps mentioned above concerning the computation and dynamics for the Lindahl equilibrium, but also bring new insights for the Fisher market equilibrium. Moreover, our duality framework also generalizes to markets with \emph{chores}, items that incur disutilities to agents. 

In what follows, we first illustrate our results in a simple setting with linear utilities, which is also the setting that motivates our initial exploration. Then in \Cref{sec:contributions}, we discuss our contributions in detail.

\subsection{Motivating Example: Lindahl Equilibrium with Linear Utilities}
We begin with a warm-up analysis of markets with linear utilities and illustrate a non-obvious equivalence between the private goods Fisher market equilibrium and the public goods Lindahl equilibrium. Specifically, we consider the setting with agents $N = \{1, 2, \ldots, n\}$ and goods $M = \{1,2, \ldots, m\}$ where each agent $i \in N$ has a linear utility $u_i(\mx) = \InAngles{\ma_i, \mx} = \sum_{j\in M} a_{ij} x_j$ where $a_{ij} \ge 0$ represents the utility of receiving a unit of good $j$. Each agent $i$ has an individual budget $B_i > 0$.

\paragraph{Private Goods and Convex Programs} \footnote{Technically speaking, these programs are maximizing a concave objective; we call them convex programs just for simplicity.} In the private goods setting, A Fisher market equilibrium is pair of price $\mp \in \-R^m_{\ge 0}$ and personalized allocations $\{\mx_i \in \-R^m_{\ge 0}\}_{i \in N}$ such that for each agent $i$, the allocation $\mx_i$ maximizes her utility subject to the budget constraint and the market clears (see \Cref{def::Fisher-market} for a formal definition).\footnote{Throughout the paper, we may use $\{\mx_i\}$ to denote $\{\mx_i\}_{i \in N}$ to simplify the notation when the context is clear.} A Fisher market equilibrium allocation with linear utilities can be computed by the celebrated Eisenberg-Gale (EG) convex program \eqref{EG-Fisher-Linear}~\citep{eisenberg1959consensus, eisenberg1961aggregation}, which maximizes the Nash social welfare $\Pi_i u_i(\mx_i)^{B_i}$. We note that the EG program can be extended to compute a Fisher market equilibrium even in the more general setting of concave and $1$-homogeneous utilities.\footnote{A function $u$ is $d$-homogeneous (homogeneous of degree $d$) if $u(c\cdot\mx) = c^d \cdot u(\mx)$ for any $c > 0$.} 
\begin{figure}[h]
  \centering

  \noindent\rule{\linewidth}{0.4pt}\vspace{0.6em}

  \begin{tabularx}{\linewidth}{@{}C C@{}}
  $\displaystyle
    \begin{aligned}
        \max_{\mx \ge 0} \,& \sum_{i \in N} B_i \log \InAngles{\ma_i, \mx_i}\\
        \text{s.t.}\,& \sum_{i \in N} x_{ij} \le 1, \quad  \text{for all } j \in M.
    \end{aligned}
  $
  &
  $\displaystyle
    \begin{aligned}
        \max_{\mb \ge 0, \mp \ge 0} \,& \sum_{i \in N, j\in M} b_{ij} \log a_{ij}
        - \sum_{j\in M} p_j \log p_j \\
        \text{s.t.}\,& \sum_{j \in M} b_{ij} = B_i, \quad \text{for all } i \in N  \\
        & \sum_{i \in N} b_{ij} = p_j,          \quad  \text{for all } j \in M \\
        &\text{\footnotesize The final allocation is obtained as $x_{ij} = \frac{b_{ij}}{p_j}$.}
    \end{aligned}
  $
  \\
  \small (\setword{EG-Fisher-Linear}{EG-Fisher-Linear}) & \small (\setword{Shmyrev-Fisher-Linear}{Shmyrev-Fisher-Linear})
  \end{tabularx}

  \vspace{0.4em}\noindent\rule{\linewidth}{0.4pt} 
\end{figure}

By convex program duality, we know that the EG program is equivalent to the Shmyrev program \eqref{Shmyrev-Fisher-Linear}~\citep{shmyrev2009algorithm}, which can be obtained by taking the dual of the dual of the EG program with reformulation.
Moreover, the proportional response dynamics~\citep{WZ2007} can be interpreted as mirror descent on the Shmyrev program and thus enjoys an $O(1/T)$ convergence to a Fisher market equilibrium~\citep{birnbaum2011distributed}.

\paragraph{Public Goods and Convex Programs} In the public good setting, a Lindahl equilibrium is a pair of public allocation and personalized prices $\{\mx, \{\mp_i\}\}$ such that each player maximizes their utility subject to the budget constraints and the price satisfies a profit-maximizing condition (see \Cref{def::Lindahl-equilibrium} for a formal definition). Lindahl equilibrium allocations with linear utilities can also be characterized as those that maximize the Nash social welfare $\Pi_i u_i(\mx)^{B_i}$~\citep{fain2016core}, which induces a public-good analog program of the EG program for private goods as shown in \ref{EG-Lindahl-Linear}.
We note that \cite{fain2016core} shows that the convex program can be extended to the more general setting of \emph{``scalar separable non-satiating"} utilities. However, separable utilities do not cover the general concave homogeneous setting as the EG program in the private goods setting, i.e., a constant elasticity of substitution (CES) utility with a negative substitution parameter, $\rho < 0$.

The convergence of proportional response dynamics for linear utilities has been discovered by different research communities~\citep{cover1984algorithm, brandl2022funding} but the rate of convergence remained open until~\cite{zhao2023convergence, kroer2025computing}. In particular, the very recent work of \cite{kroer2025computing} gives a Shmyrev-type program for \ref{EG-Lindahl-Linear} using convex program duality.  They show that the proportional response dynamics for computing a Lindahl equilibrium is equivalent to running mirror descent over the \ref{Shmyrev-Lindahl-Linear} and enjoys $O(1/T)$ convergence.

\begin{figure}[h]
  \centering
  \noindent\rule{\linewidth}{0.4pt}\vspace{0.6em}
  \begin{tabularx}{\linewidth}{@{}C C@{}}
  $\displaystyle
    \begin{aligned}
        \max_{\mx \ge 0}\,&\, \sum_{i \in N} B_i \log \InAngles{\ma_i, \mx}\\
        \text{s.t.}\,&\, \sum_{j \in M} x_{j} \le B.
    \end{aligned}
  $
  &
  $\displaystyle
    \begin{aligned}
        \max_{\mb \ge0,\mx\ge 0}\,&\, \sum_{i \in N, j\in M} b_{ij} \log a_{ij} - \sum_{i \in N, j\in M} b_{ij} \log \frac{b_{ij}}{x_j} \\
        \text{s.t.} \,&\,\sum_{j \in M} b_{ij} = B_i, \quad \text{for all } i \in N  \\
        &\,\sum_{i \in N} b_{ij} = x_j,             \quad  \text{for all } j \in M
    \end{aligned}
  $
  \\
  \small (\setword{EG-Lindahl-Linear}{EG-Lindahl-Linear}) & \small (\setword{Shmyrev-Lindahl-Linear}{Shmyrev-Lindahl-Linear})
  \end{tabularx}
  \vspace{0.4em}\noindent\rule{\linewidth}{0.4pt} 
\end{figure}

\paragraph{Distinctions between Private Goods and Public Goods} As emphasized in \citep{kroer2025computing}, market equilibrium for private goods and public goods are very different. This is evident from the Shmyrev programs for private goods and public goods, which have notable differences in structures and properties. 
\begin{itemize}
    \item[1.] \ref{Shmyrev-Fisher-Linear} has variables corresponding to prices, which do not appear in \ref{Shmyrev-Lindahl-Linear}. Similarly, \ref{Shmyrev-Lindahl-Linear} has variables corresponding to allocations, which do not appear in \ref{Shmyrev-Fisher-Linear}.
    \item[2.] \ref{Shmyrev-Fisher-Linear} always has rational solutions~\citep{devanur2008market, vazirani2012notion}, while \ref{Shmyrev-Lindahl-Linear} may have only irrational solutions~\citep[Example 3]{kroer2025computing}.
    \item[3.] The nonlinear term in the objective is an unusual normalized entropy function in \ref{Shmyrev-Lindahl-Linear}, whereas it is the usual entropy function on the prices in \ref{Shmyrev-Fisher-Linear}.
\end{itemize}
These differences seem to suggest a fundamental difference between the Fisher market equilibrium and the Lindahl equilibrium.

\paragraph{Equivalence by Duality between Linear and Leontief Utilities} However, our key initial observation is that public goods markets with linear utilities are equivalent to private goods markets, not with linear utilities, but with \emph{Leontief utilities}.  While a linear utility represents fully substitute goods, a Leontief utility represents fully complementary goods. Given $\ma_i \ge 0$, a Leontief utility has the form $u_i(\mx) = \min_{j \in M} \{\frac{x_j}{a_{ij}}\}$, which means the agent receives $\ge 1$ utility only when she receives enough amount of every good, i.e., $x_j \ge a_{ij}$ for all $j$.

We first present the Shmyrev program for private goods Fisher markets with Leontief utilities. 
\begin{equation}\tag{Shmyrev-Fisher-Leontief}
    \begin{aligned}
        \max_{\mb \ge0, \mp \ge 0} &\sum_{i \in N, j\in M} b_{ij} \log a_{ij} - \sum_{i \in N, j\in M} b_{ij} \log \frac{b_{ij}}{p_j} \\
        \text{s.t. }& \sum_{j \in M} b_{ij} = B_i, \quad \text{for all } i \in N  \\
        &\sum_{i \in N} b_{ij} = p_j,             \quad  \text{for all } j \in M
    \end{aligned}
    \label{Shmyrev-Fisher-Leontief}
\end{equation}
When we take a closer look at \ref{Shmyrev-Fisher-Leontief} and \ref{Shmyrev-Lindahl-Linear}, we find these two programs are the same program up to a change of variables between $\mx$ and $\mp$. Thus if $(\mb, \mp)$ is an optimal solution of \ref{Shmyrev-Fisher-Leontief}, then let $\mx = \mp$, we must have $(\mb, \mx)$ is an optimal solution of \ref{Shmyrev-Lindahl-Linear}. In other words, the roles of allocations and prices are exchanged: an equilibrium price vector in Fisher markets with Leontief utilities corresponds to a public goods Lindahl equilibrium allocation with linear utilities. We remark that the equivalence between the two programs has several immediate implications.
\begin{itemize}
    \item This gives another explanation for why \ref{Shmyrev-Lindahl-Linear} may have irrational solutions since \ref{Shmyrev-Fisher-Leontief} could have irrational solutions \cite{codenotti2004efficient}.
    \item 
    Using the equivalence between the \ref{Shmyrev-Lindahl-Linear} and \ref{Shmyrev-Fisher-Leontief}, we can recover the recently established $O(1/T)$ convergence of proportional response dynamics for Lindahl equilibrium with linear utilities~\citep{kroer2025computing},  as it has been established that proportional responses dynamics has $O(1/T)$ convergence rate for Fisher market equilibrium with Leontief utilities~\citep{cheung2018dynamics}.
    \item T\^atonnement is a classic price adjustment procedure that computes a market equilibrium. While t\^atonnement for private goods has been extensively studied, t\^atonnement for public goods remains unexplored. \cite{kroer2025computing} discussed a possible path to t\^atonnement for public goods through the new \ref{Shmyrev-Lindahl-Linear} program, but does not conclude any convergence rates and explicitly leaves it as an open question. Using the equivalence to the private goods setting and existing results, we can easily get the update rule and convergence rates of t\^atonnement for public goods. In \Cref{sec::tatonnement}, we present the  convergence of t\^atonnement in a more general setting. 
\end{itemize}

From the above implications, we can see that the new perspective --- viewing a linear public goods market as a Leontief private goods market --- provides a formal bridge for transferring analytical tools and computational results between the two settings. The counterintuitive nature of this specific correspondence (substitutes $\leftrightarrow$ complements) motivates the central question of this paper: \emph{What is the general, underlying connection between Fisher and Lindahl equilibria?}

In the following sections, we develop a duality framework that generalizes these observations. This framework formally unifies these two equilibrium models. We demonstrate that this duality is a powerful tool, not only for extending results from the private to the public goods setting but also for generating new insights for private goods and even the private/public chores setting.

\subsection{Our Contributions}\label{sec:contributions}
The motivating example suggests a deeper connection that is obscured by the superficial differences between the Lindahl and Fisher models. This paper's main contribution is to formalize and generalize this connection by introducing a comprehensive duality framework that unifies these two seemingly disparate equilibrium concepts.

Our framework builds upon a \emph{dual utility} function, $\tilde{u}$, which we derive from the standard indirect utility function $v(\mp, B)$ (specifically, $\tilde{u}(\mx) = 1/v(\mx, B)$).\footnote{Given a utility function $u$, the indirect utility $v(\mp, B)$ is the maximum utility attainable under the price $\mp$ and budget constraint $B$. See \Cref{definition:indirect utility} for a formal definition.} Using this concept, we establish a precise equivalence (\Cref{thm::main-duality}) between the two market types. \emph{A public goods market $\mathcal{L}$ with utilities $\{u_i\}$ admits a Lindahl equilibrium $(\mx, \{\mp_i\})$ if and only if a corresponding dual private goods market $\mathcal{F}$ with dual utilities $\{\tilde{u}_i\}$ admits a Fisher market equilibrium $(\{\tilde{\mx}_i\}, \tilde{\mp})$.} This duality is characterized by a remarkable exchange in the roles of prices and allocations: the Lindahl equilibrium allocation $\mx$ becomes the Fisher market price $\tilde{\mp}$, while the personalized Lindahl prices $\{\mp_i\}$ become the personalized Fisher allocations $\{\tilde{\mx}_i\}$.

This framework is not merely a theoretical curiosity; it serves as a powerful methodological bridge, allowing us to transfer deep insights and powerful algorithms from the well-studied Fisher setting to the less-explored Lindahl domain. 

For example, our framework immediately yields a generalized Shmyrev program (\ref{Shmyrev-Lindahl-CES}) for Lindahl equilibrium with CES utilities. This program encompasses the linear utility case from \cite{kroer2025computing} as a special instance. Furthermore, for the Leontief utility case, \cite{brandt2023balanced} previously analyzed a best-response dynamic, proving its convergence with a potential function. Interestingly, we find this potential function is precisely equivalent to our Shmyrev program when specialized to the Leontief case.

Furthermore, we also demonstrate the power of this duality through several key applications:

\paragraph{Nash Welfare and the Lindahl Equilibrium} We give a characterization of Lindahl equilibria for public goods markets with general concave and $1$-homogeneous utilities using the Nash Social Welfare (NSW) maximization program ($\sum_i B_i \log u_i(\mathbf{x}) \text{ s.t. } \mathbf{x}^\top \mathbf{1} \leq \sum_i B_i$). This is an analog to the goods setting where the Fisher market equilibrium is characterized by the Eisenberg-Gale (EG) NSW maximization program. Our proof naturally follows from our duality framework by considering the dual of the EG convex program for the dual Fisher market. We remark that previous works prove this result only for special cases of utilities, including separable homogeneous~\cite{fain2016core}, linear~\citep{fain2016core,10.1093/restud/rdaf043}, and Leontief utilities~\citep{brandt2023balanced}

For general concave utilities that are not necessarily homogeneous, a Lindahl equilibrium allocation may attain the optimal NSW. Nevertheless, we show that any Lindahl equilibrium allocation achieves a $(1/e)^{1/e}$ approximation of the optimal NSW (\Cref{theorem::non-homogenous}), and we prove this bound is tight (see \Cref{example::etoe}). This result extends previous results for Fisher market equilibria~\citep{garg2025approximating} through our duality framework.

\paragraph{Unifying and Interpreting Market Dynamics} Our framework provides a dual correspondence for analyzing market dynamics. We demonstrate that our duality framework enables us to transfer classic dynamics for market equilibrium, including proportional response dynamics (PRD) and t\^atonnement, from the private goods Fisher markets to the public goods markets. This not only gives new dynamics for the Lindahl equilibrium with convergence guarantees, but also gives new insight into the dynamics for the Fisher market equilibrium. 

We first analyze proportional response dynamics. For the general class of CES utility functions, their dual utility functions are also CES utility functions (\Cref{ex:CES and indirect utility}). We propose PRD for the Lindahl equilibrium and show that they are equivalent to mirror descent on \ref{Shmyrev-Lindahl-CES}. We then obtain similar convergence guarantees (\Cref{theorem:Lindahl-CES-PRD-complement,theorem:Lindahl-CES-PRD-substitute}) to those of the Fisher market equilibrium using duality.

A key insight from our duality concerns the interpretation of a complex expression within the PRD for Fisher markets with general gross substitutes utilities~\citep{cheung2025proportional}. By applying Roy's identity, we demonstrate that this term has a simple and intuitive interpretation in the dual Lindahl equilibrium: it represents the \emph{'best response expenditure'} (for details, see the paragraph 'A Dual Interpretation via Roy's Identity' in \Cref{sec:prd}). This insight, in turn, allows us to derive a novel dynamic (\Cref{theorem:PRD:lindahl:tc}) for Lindahl equilibria with \emph{total complements} utilities (\Cref{def:tc}) - the dual class of gross substitutes. This total complements class notably includes CES utilities with elasticity parameter $\rho \in (-\infty, 0)$.

Furthermore, we establish the convergence of PRD for Lindahl markets with gross substitutes utilities (\Cref{thm:prd-lindahl-gs}), which generalizes the PRD for the linear case \cite{kroer2025computing}. Additionally, we can map this convergent dynamic back to the Fisher setting via our duality, and it generates a new corresponding dynamic for Fisher markets with total complements utilities (\Cref{thm:fisher:tc}), which includes CES utilities with $\rho \in (-\infty, 0)$. This result closes the gap that, while the PRD dynamics with CES utilities with negative $\rho$ was previously derived using mirror descent on the generalized Shmyrev program, its convergence could not be established by generalizing existing results from the Fisher market with gross substitutes case. Our framework demonstrates that the correct generalization path proceeds from the Lindahl equilibrium with gross substitutes utilities through the dual correspondence. Our results for PRD are summarized in \Cref{tab:PRD}.

\begin{table}[t]
\centering
\resizebox{\columnwidth}{!}{%
\begin{tabular}{@{}lcccc@{}}
\toprule
\textbf{Markets / Utilities} &
  \textbf{Linear} &
  \textbf{CES} &
  \textbf{Gross Substitutes} &
  \textbf{Total Complements} \\
\midrule
Private Goods Fisher Equilibrium &
  \citep{WZ2007, birnbaum2011distributed} &
  \citep{cheung2018dynamics} &
  \citep{cheung2025proportional} &
  \Cref{thm:fisher:tc} \\
Public Goods Lindahl Equilibrium &
  \citep{kroer2025computing} &
  \Cref{theorem:Lindahl-CES-PRD-complement,theorem:Lindahl-CES-PRD-substitute} &
  \Cref{thm:prd-lindahl-gs} &
  \Cref{theorem:PRD:lindahl:tc} \\
\bottomrule
\end{tabular}%
}
\caption{Convergence results on proportional response dynamics for market equilibrium in private goods markets and public goods markets with different types of utilities.}
\label{tab:PRD}
\end{table}

Finally, we illustrate the broad applicability of our framework to \emph{t\^atonnement} dynamics. As an example, we show that the t\^atonnement process for nested-CES utilities in a Fisher market can be transformed into a new dynamic in the dual Lindahl equilibrium (\Cref{thm:tat:lindahl}). This new Lindahl dynamic possesses an intuitive economic interpretation, where the allocations of public goods are adjusted based on agents' overpayments.   

\paragraph{Markets with Chores} We show that our duality framework works for markets with chores, items that incur disutilities in \Cref{sec::chores}. 

First, we address an open problem in Fisher markets for private chores. A Fisher market equilibrium for chores, also known as a competitive equilibrium for chores, is a pair of prices and personalized allocations that (1) each agent receives the disutility-minimizing bundle of chores subject to an earning constraint; (2) every chore is fully allocated (see \Cref{definition:Fisher-chores}). This framework is motivated by practical fair chore division problems, such as teachers dividing teaching loads or roommates dividing household chores. For the class of convex and 1-homogeneous disutility functions, unlike the goods setting, the set of competitive equilibra is characterized by the Karush–Kuhn–Tucker (KKT) points of a Nash Welfare \emph{minimization} program~\citep{bogomolnaia2017competitive} where each agent received non-zero disutilities. However, these optimization formulation for chores has ``poles'' --- infeasible points on the boundary of the feasible region that attract iterative methods and drive the objective to negative infinity~\citep{boodaghians2022polynomial,chaudhury2022competitive}.  \citet{chaudhury2024competitive} give a pole-free formulation for the special case of linear disutilities, but the general convex and 1-homogeneous setting remained open.

We resolve this question and propose a new optimization formulation for general convex and $1$-homogeneous. Our main result (\Cref{theorem: fisher-chores-KKT}) establishes that the KKT points of this new, pole-free program correspond exactly to the Fisher market equilibrium for chores. This opens the possibility of applying practical first-order methods for computing competitive equilibria. On the technical side, we introduce the indirect disutility function and the dual disutility function (the inverse of the indirect disutility function) similar to the goods setting. Our program is based on the dual disutility and the connection to the Lindahl equilibrium in the dual market. We establish Roy's Identity-like characterization of the optimal demand for chores, which is crucial when extending from the linear setting to the general convex and $1$-homogeneous setting.

Second, we explore the public chores setting where each agent must receive the same allocation of public chores. Allocations of public chores arise in the practice. Consider a community with several waste-processing facilities and a fixed amount of total waste. Individuals have preferences over how it is allocated --- for example, they prefer nearby facilities to handle less waste. In this case, each facility is a chore and the problem is how to fairly divide the total waste among the facilities. We propose a new solution concept we call Lindahl equilibrium for public chores via duality to the Fisher market for private chores. We show that the Lindahl equilibrium allocation is (weakly) Pareto-optimal. Moreover, we establish a duality (\Cref{theorem:chores-duality}) between the Lindahl equilibrium for public chores and the Fisher market for private chores, analogous to the duality for goods. Leveraging this new duality, we demonstrate that for the class of convex and 1-homogeneous disutility functions, the Lindahl equilibrium also corresponds to the KKT points of a Nash Welfare minimization program.

\paragraph{Roadmap} The remainder of the paper is organized as follows. \Cref{sec:preliminaries} provides preliminaries.  \Cref{sec::dual::goods} introduces the duality framework.   \Cref{sec:homogeneous:nsw} analyzes the relationship between the Lindahl equilibrium and Nash Social Welfare.  \Cref{sec::marketdynamics} explores the market dynamics for both private goods and public goods using the duality framework. \Cref{sec::chores} analyzes markets with chores, considering cases for both Fisher markets and the Lindahl equilibrium.  \Cref{sec::missing} contains all deferred proofs and additional discussions.

\subsection{Related Literature}\label{sec:related works}
\paragraph{Fisher Markets and Nash Social Welfare} Competitive equilibrium stands as a central solution concept in economics, with classical results establishing its existence in the Fisher model and more general exchange market models \cite{arrow1954existence, mckenzie1954equilibrium}.  The Fisher market in particular has been the subject of extensive investigation, including its equilibrium properties \cite{megiddo2007continuity, chen2012incentive, bei2019earning, goktas2021consumer}, computational algorithms \cite{codenotti2004efficient, Orlin10, nesterov2018computation, kroer2019computing, garg2022approximating}, and numerous extensions \cite{branzei2014fisher, Cole:2016:LMG:2940716.2940720, liaostatistical, conitzer2022pacing, gao2023infinite}. While existence is guaranteed under broad conditions, the task of computing an equilibrium is known to be computationally challenging \cite{chen2009spending, CDDT2009, CPY2017, CSVY2006, garg2017settling, VazYan10, bei2019earning, deligkas2024constant}. 

A fundamental connection between Fisher market equilibrium and Nash Social Welfare (NSW) was established by \cite{eisenberg1959consensus, eisenberg1961aggregation}. They demonstrated that maximizing the NSW \cite{nash1950bargaining} yields a competitive equilibrium under general concave utility functions that are 1-homogeneous. Recent work has also focused on the connection between Fisher market equilibrium and NSW for general concave utilities \cite{garg2024satiation, garg2025approximating}. The maximum Nash Social Welfare (NSW) is itself a classical objective for allocating goods and has received significant attention within the social choice and fair division literature, even extending to chore allocation \cite{moulin2004fair, bogomolnaia2017competitive, CDGJMVY2017}. Furthermore, NSW provides remarkable fairness guarantees in settings with indivisible goods, where it is known to produce allocations that satisfy envy-freeness up to one good (EF1) under additive valuations \cite{CMP016}. 

\paragraph{Lindahl Equilibrium} The Lindahl equilibrium, introduced by \cite{foley1970lindahl} based on Lindahl's \cite{lindahl1958just} ideas of personalized taxation, has garnered increasing attention within recent studies on fair division and public goods\cite{conitzer2017fair}. This interest is primarily motivated by its strong proportional fairness properties, particularly its connection to the core \cite{airiau2023portioning, bogomolnaia2005collective, fain2016core, brandl2022funding, brandt2023balanced, 10.1093/restud/rdaf043}. Recent work has also extended the classical Lindahl equilibrium model to accommodate settings with discrete choices \cite{peters2021market, munagala2022auditing,nguyen2025few,}. 

\paragraph{Connection between Fisher market and Lindahl Equilibrium} 
We note that a public goods market can be converted to a private goods market by a simple reduction: we just duplicate $n$ copies of each good such that each of them is only valuable for one agent, and then add additional constraints that the consumption of these goods must be equal among all agents. This reduction explicitly enforces that the allocation of public goods is the same for each agent, and is used in~\citep{foley1970lindahl} to prove the existence of a Lindahl equilibrium by invoking a fixed-point argument~\citep{debreu1962new}. Similar reductions that convert a public goods market to a private goods market by introducing copies of goods are also explored in \cite{garg2021markets}. Our duality-based reduction from the Lindahl equilibrium to the Fisher market equilibrium is quite different, as we do not introduce duplicated goods or add constraints. Instead, the duality-based reduction uses the duality between indirect and direct utility functions, and the correspondence between two markets holds through exchanging the roles of allocations and prices. During the preparation of this paper, we found that \cite{ruys1972existence} uses a similar duality argument to prove the existence of a Lindahl equilibrium in markets with public goods only, although not through the language of indirect utility functions.

\paragraph{Dynamics in the Market}
The study of market dynamics is a fundamental topic in both economics and algorithmic game theory. A foundational dynamic is the t\^atonnement process, introduced by \citet{Walras1874}, which describes a natural market adjustment driven by supply and demand, where prices rise for overdemanded goods and fall for underdemanded ones. Following \citet{samuelson1983foundations}'s continuous model, seminal works by \citet{ArrowHurwicz1958} and \citet{ABH1959} introduced the gross substitutes property as a crucial condition to ensure its convergence. More recently, the intersection of computer science and economics has prompted the exploration of discrete variants of t\^atonnement algorithms \cite{CMV2005, ARY14, goktas2023tatonnement, CF2008, CCR2012, CCD2013, cole2019balancing, goktas2021consumer, dvijotham2022convergence, gao2020first, fleischer2008fast, CheungCole2018async, nan2025convergence}. Another widely studied dynamic is proportional response dynamics (PRD), noted for its simplicity in networked markets \cite{LLSB08}. \citet{WZ2007} were the first to formally identify the power of PRD in driving linear Arrow-Debreu markets toward competitive equilibrium. Since then, a growing body of research has extended PRD's convergence guarantees to Arrow-Debreu markets, Fisher markets,
production economies and attention markets, demonstrating its robustness across a range of settings.~\cite{Zhang2011,birnbaum2011distributed,cheung2018dynamics,BranzeiMehtaNisan2018NeurIPS,CheungHN19,gao2020first,BranzeiDevanurRabani2021,CheungLP2021,ZhuCX2023WWW,kolumbus2023asynchronous,li2024proportional,cheng2024tight}. Other relevant market algorithms include various ascending-price auction mechanisms \cite{GargKap2006, garg2006price, garg2007market, garg2004auction, bei2019ascending, garg2023auction}. We remark that works on dynamics for Lindahl equilibria in public goods market are very much sparser compared to Fisher markets. \citet{brandt2023balanced} study a best-response type dynamics in the setting of Leontief utilities and show asymptotic convergence, while \citet{kroer2025computing} study PRD for linear utilities and prove $O(1/T)$ convergence. Convergence of both t\^atonnement and PRD for the Lindahl equilibrium in more general settings remained unknown before our work.

\paragraph{Chores} The study of competitive equilibrium (CE) for chores, introduced in the foundational work of Bogomolnaia et al. \cite{bogomolnaia2017competitive}, has become an active area of research \cite{liu2024mixed, amanatidis2022fair}. Much of the work in this area has focused on computational aspects. Br\^anzei and Sandomirskiy \cite{branzei2024algorithms} established polynomial-time algorithms for computing an exact CE when the number of agents or chores is constant. This was extended by Garg and McGlaughlin \cite{garg2020computing} to mixed settings involving both goods and chores. For the general case, Boodaghians et al. \cite{boodaghians2022polynomial} and Chaudhury et al. \cite{chaudhury2022competitive} developed a polynomial-time algorithm for computing approximate CE, where their iterations require solving non-linear convex programs. Recently, \citet{chaudhury2024competitive} demonstrated a correspondence between the CE for chores with linear utilities and the KKT points of the Nash Social Welfare minimization problem. Furthermore, the model of chores has been adapted to related domains, such as the public chore setting \cite{elkind2025not}.

\section{Preliminaries}\label{sec:preliminaries}
\paragraph{Basic Notations} We use $\boldsymbol{0}$ and $\boldsymbol{1}$ to denote the all-zero vector and all-one vector, respectively; their dimensionality will be clear from the context.
For two vectors $\mx, \my \in \-R^m$, $\mx \ge \my$ means $x_j \ge y_j$ for all $j \in [m]$, while $\mx \gg \my$ means $x_j > y_j$ for all $j \in [m]$.

\paragraph{(Quasi)-Convexity/(Quasi)-Concavity} 
A function $f: \+S \rightarrow \-R$ defined on a convex set $S \subseteq \-R^m$ is \emph{convex} if for any $x, y \in \+S$ and $\lambda \in [0,1]$, we have $f(\lambda x + (1-\lambda)y)\le \lambda f(x) + (1 - \lambda) f(y)$. A function $f$ is \emph{concave} if $-f$ is convex. A function $f: \+S \rightarrow \-R$ defined on a convex set $S \subseteq \-R^m$ is \emph{quasi-convex} if for every $x, y \in \+S$ and $\lambda \in [0,1]$, we have $f(\lambda x+(1-\lambda) y) \le \max\{f(x), f(y)\}$. An equivalent definition of a quasi-convex function is that every sublevel set $\{x: f(x) \le \alpha\}$ is a convex set. A function $f$ is \emph{quasi-concave} if $-f$ is quasi-convex, or equivalently, every superlevel set $\{x: f(x) \ge \alpha\}$ is convex.

\paragraph{Utility Functions} A utility function $u: \-R^m_{\ge 0} \rightarrow \-R_{\ge 0}$ is \emph{non-decreasing} if for any $\mx, \my \in \-R^m_{+}$, $\mx \ge \my$ implies $u(\mx) \ge u(\my)$. A function $u$ is \emph{increasing} or \emph{monotone} if for any $\mx, \my \in \-R^m_{+}$, $\mx \gg \my$ implies $u(\mx) > u(\my)$. A utility function $u$ is \emph{strictly increasing} if for any $\mx, \my \in \-R^m_{+}$, $\mx \ge \my$ and $\mx \ne \my$ implies $u(\mx) > u(\my)$. A utility function is \emph{locally nonsatiated} if for every $\mx \in \-R^m_{\ge 0}$ and $\varepsilon > 0$, there is $\my  \in \-R^m_{\ge 0}$ such that $\InNorms{\mx - \my} \le \varepsilon$ and $u(\my) > u(\mx)$. The relationship between these properties is as follows.
\begin{align*}
    \text{strictly increasing} \implies \text{increasing = monotone} \implies \begin{cases}
\text{non-decreasing}\\
\text{locally nonsatiated}
\end{cases}
\end{align*}
Most of our results hold under the following assumption on the utility function
\begin{assumption}\label{assumption:c-nd-ln-qc}
    The utility function $u: \-R^m_{\ge 0} \rightarrow \-R_{\ge 0}$ is continuous, non-decreasing, locally nonsatiated, and quasi-concave.
\end{assumption}

\subsection{Markets with Private / Public goods}
\paragraph{Private Goods: Fisher Market Equilibrium}
A Fisher market, $\+F = (M, N, \{u_i\}, \{B_i\})$, consists of $m$ divisible private goods $M=\{1,2,\ldots,m\}$ and $n$ agents $N= \{1,2, \ldots,n\}$.  We assume, without loss of generality, a unit supply of each good. Each agent $i \in \{1, \dots, n\}$ has an initial budget $B_i > 0$ and a utility function $u_i: \mathbb{R}_{\geq 0}^m \to \mathbb{R}_{\geq 0}$. The goal is to find a Fisher market equilibrium, also known as \emph{competitive equilibrium}, which consists of a public price vector $\mp \in \-R^m_{\ge 0}$ with $p_j \ge 0$ being the price for good $j$, and personalized allocations $\mx_i \in \-R^m_{\ge 0}$ for every agent $i$ with $x_{ij}$ being the amount of good $j$ allocated to agent $i$.

\begin{definition}[Fisher market equilibrium]\label{def::Fisher-market} Given $\+F$, let $\{\mx_i\}$ be a personalized allocations and $\mp$ be a price vector. Then $(\{\mx_i\}, \mp)$ is a \emph{Fisher market equilibrium} if
\begin{enumerate}[label=(\roman*)]
    \item $\mx_i$ is affordable: $\InAngles{\mp, \mx_i} \le B_i$ for every agent $i \in N$,
    \item $\mx_i$ is utility-maximizing: $\mx_i\in \argmax_{\my_i \in \-R_{\geq 0}^m: \InAngles{\mp, \my_i} \leq B_i} u_i(\my_i)$ for every agent $i \in N$,
    \item \emph{market clears}: 
    for every $j \in M$, $\sum_i x_{ij} \le 1$; and whenever $p_j > 0$, then $\sum_i x_{ij} = 1$.
\end{enumerate}
\end{definition}
\paragraph{Public Goods: Lindahl Equilibrium} A public good market $\+L = (M, N, \{u_i\}, \{B_i\})$ consists of $m$ public goods $M=\{1,2,\ldots,m\}$ and $n$ agents $N= \{1,2, \ldots,n\}$. Each agent $i \in \{1, \dots, n\}$ has an initial budget $B_i > 0$ and a utility function $u_i: \mathbb{R}_{\geq 0}^m \to \mathbb{R}_{\geq 0}$.  The goal is to find a Lindahl equilibrium~\citep{foley1970lindahl} that consists of a public allocation $\mx \in \-R^m_{\ge 0}$ and personalized price vectors $\{\mp_i\}$ that satisfies the following conditions. Here $x_j \ge 0$ represents the amount of budget spent on good $j$ where $p_{ij}$ is the price of good $j$ for agent $i$.

\begin{definition}[Lindahl equilibrium] \label{def::Lindahl-equilibrium} Given $\+L$, let $\mx$ be an allocation and $\{\mp_i\}$ be nonnegative personalized prices. Then $(\mx, \{\mp_i\}_{i \in N})$ is a \emph{Lindahl equilibrium} if 
\begin{enumerate}[label=(\roman*)]
    \item $\mx$ is \emph{affordable}: $\InAngles{\mp_i, \mx} \leq B_i$ for every agent $i \in N$,
    \item $\mx$ is \emph{utility-maximizing}: $\mx \in \argmax_{\my \in \-R_{\geq 0}^m: \InAngles{\mp_i, \my} \leq B_i} u_i(\my)$ for every agent $i \in N$,
    \item $\mx$ is \emph{profit-maximizing}: for every $j \in M$, $\sum_i p_{ij} \leq 1$; and whenever $x_j >0$ then $\sum_i p_{ij} = 1$.
\end{enumerate}
\end{definition}
The profit-maximizing condition has its name because it is equivalent to the condition that a producer of public goods maximize his profit subject to unit cost of each good, that is, $\mx\in \argmax_{\my \in \-R^m_{\ge 0}} \sum_{i}\sum_j p_{ij} y_j - \sum_j y_j$ (see e.g., \citep[Lemma 1]{kroer2025computing} for a proof). A discussion of the case where the producer has a general cost function can be found in \Cref{sec::general-lindahl-discussion}.

\section{Duality between Fisher Market Equilibrium and Lindahl Equilibrium} \label{sec::dual::goods}
\paragraph{Indirect Utility Function} Our duality framework is built upon the classical concept of the indirect utility function.
\begin{definition}[Indirect Utility Function]\label{definition:indirect utility}
    Given a utility function $u: \-R^m_{\ge 0} \rightarrow \-R_{\ge 0}$, its indirect utility function $v: \-R^m_{\ge 0} \times \-R_{> 0} \rightarrow \-R_{\ge 0}$ is defined as
    \begin{align*}
        v(\mp, B) = \sup_{\mx \in \-R^m_{\ge 0}, \InAngles{\mp, \mx} \le B} u(\mx).
    \end{align*}
    That is, $v(\mp, B)$ is the maximal utility attainable under the price $\mp$ and budget constraint $B$.
\end{definition} 
The indirect utility function gives a dual characterization of the utility function as shown in the following theorem. Different versions of this theorem have been proved under various assumptions on the utility function in the literature. We present a general version that holds under \Cref{assumption:c-nd-ln-qc} and include the proof in \Cref{sec::direct-indirect}.
\begin{theorem}[Duality between direct and indirect utility]\label{theorem:duality-utility-indirec-utility}
    Suppose $u: \-R^m_{\ge 0} \rightarrow \-R_{\ge 0}$ satisfies \Cref{assumption:c-nd-ln-qc}. Then for any $B > 0$ and all $\mx\in \-R^m_{\ge 0} \setminus\{\boldsymbol{0}\}$ it holds that 
    \begin{align*}
        u(\mx) = \inf_{\mp \in \-R^{m}_{\ge 0}: \InAngles{\mx, \mp} \le B} v(\mp, B).
    \end{align*}
\end{theorem}

Our motivating examples, linear and Leontief utilities, illustrate this concept. For a linear utility $u_i(\mx) = \sum_j a_{ij} x_{ij}$, the indirect utility is $v_i(\mp, B_i) = B_i \max_j\left\{\frac{a_{ij}}{p_j}\right\}$. If we then define a new function $\dual{u}_i(\mx) \triangleq 1/v_i(\mx, B_i)$ - substituting the price vector $\mp$ with the quantity vector $\mx$ in the functional form of $v_i$ - we obtain $\dual{u}_i(\mx) =\frac{1}{B_i} \min_j \left\{\frac{x_j}{a_{ij}}\right\}$, which is precisely a Leontief utility.\footnote{We remark that for any scalar $c > 0$, utility $u$ and utility $c u$ represent the same preference. Moreover, positive scaling of any utility function has no effect on the equilibrium of markets with private or public goods.} Conversely, starting with a Leontief utility $u_i(\mx) = \min_j \left\{ \frac{x_j}{a_{ij}} \right\}$, applying the same transformation $\dual{u}_i(\mx) \coloneqq 1/v_i(\mx, B_i)$ yields $\dual{u}_i(\mx) = \frac{1}{B_i} \sum_j a_{ij} x_{ij}$, a linear utility.

This reciprocal relationship is not limited to these two specific cases. The family of Constant Elasticity of Substitution (CES) utilities provides a broader illustration of this duality and serves as the motivation for our central definition.

\begin{example}[Constant Elasticity of Substitution (CES) Utility and Indirect Utility]\label{ex:CES and indirect utility} A CES utility with elasticity of substitution $\rho \in (-\infty, 1)$ is defined as
$$ u_i(\mx) = \left(\sum_{j=1}^m a_{ij} x_{j}^{\rho_i}\right)^{1/\rho_i}, \quad \text{with } \rho_i \in (-\infty, 1). $$

For a CES utility with $\rho_i$, its indirect utility is 
$$ 
v_i(\mp, B_i) = B_i \left(\sum_{j=1}^m a_{ij}^{1 - \dual{\rho}_i} p_j^{\dual{\rho}_i}\right)^{-1/\dual{\rho}_i}, \quad \text{where } \dual{\rho}_i = \frac{\rho_i}{\rho_i - 1}. 
$$
We note that the indirect utility function itself looks like a CES utility. In fact, if we look at the function $\dual{u}(\mx) = 1/v(\mx, B_i)$ for a fixed $B_i$, we get
$$\dual{u}(\mx) = \frac{1}{v_i(\mx, B_i)} = \frac{1}{B_i} \left(\sum_{j=1}^m a_{ij}^{1 - \dual{\rho}_i} x_j^{\dual{\rho}_i}\right)^{1/\dual{\rho}_i}, \quad \text{where } \dual{\rho}_i = \frac{\rho_i}{\rho_i - 1}. $$
\end{example}

\paragraph{Dual Utility} The correspondence between the linear utility and the Leontief utility motivates the following definition of \emph{dual utility} using the indirect utility function.
\begin{definition}[Dual Utility] \label{def::dual-utility-goods}
    Given a utility function $u$ and budget $B > 0$, its \emph{dual utility} is $\dual{u}$ such that $\dual{u}(\mx) = \frac{1}{v(\mx, B)}$, where $v$ is the indirect utility of $u$.
\end{definition}
Our main result is a duality characterization between the Lindahl equilibrium for public goods and the Fisher market equilibrium for private goods using dual utilities. Specifically, we show that under mild assumptions, every Lindahl equilibrium of a public goods market corresponds to a Fisher market equilibrium of a \emph{dual} private goods market, with the same agents, goods, and budgets, but with the dual utilities. The correspondence works by exchanging the role of allocations and prices between the two markets: for any Lindahl equilibrium $(\mx, \{\mp_i\})$, the Lindahl equilibrium allocation $\mx$ is the Fisher market equilibrium price $\dual{\mp}$, while the Lindahl equilibrium personalized prices $\{\mp_i\}$ are the Fisher market personalized allocations $\{\dual{\mx}_i\}$. 

\begin{theorem}[Duality between Lindahl Equilibrium and Fisher Market Equilibrium] \label{thm::main-duality} Consider a public goods market $\+L = (N, M, \{u_i\}, \{B_i\})$ where $\{u_i\}$ satisfy \Cref{assumption:c-nd-ln-qc}. Define the dual private goods market $\+F = (N, M, \{\dual{u}_i\}, \{B_i\})$ where $\dual{u}_i$ are the dual utilities of $\{u_i\}$. Then $(\mx, \{\mp_i\})$ is a Lindahl equilibrium of $\+L$ if and only if $(\{\dual{\mx}_i\}, \dual{\mp})$ is a Fisher market equilibrium of $\+F$, where $\dual{\mx}_i = \mp_i$ and $\dual{\mp} = \mx$. 
\end{theorem}
\begin{remark}
    The other direction of the statement also holds with a similar proof. Consider a private goods Fisher market $\+F = (N, M, \{u_i\}, \{B_i\})$ where $\{u_i\}$ satisfy \Cref{assumption:c-nd-ln-qc}. Define the dual public goods market $\+F = (N, M, \{\dual{u}_i\}, \{B_i\})$ where $\dual{u}_i$ are the dual utilities of $\{u_i\}$. Then $(\{\mx_i\}, \mp)$ is a Fisher market equilibrium of $\+F$ if and only if $(\dual{\mx}, \{\dual{\mp}_i\})$ is a Lindahl market equilibrium of $\+F$, where $\dual{\mx} = \mp_i$ and $\dual{\mp}_i = \mx_i$. 
\end{remark}
\begin{proof} The core of \Cref{thm::main-duality} is to establish the following equivalent definition of Lindahl equilibrium (\Cref{def::Lindahl-equilibrium}).
\begin{definition}[Equivalent Definition of Lindahl Equilibrium]\label{def:Lindahl-equivalent}
 Given $\+L$, let $\mx$ be an allocation and $\{\mp_i\}$ be nonnegative personalized prices. Then $(\mx, \{\mp_i\})$ is a \emph{Lindahl equilibrium} if 
\begin{enumerate}[label=(\roman*)]
     \item $\mx$ is \emph{affordable}: $\InAngles{\mp_i, \mx} \leq B_i$ for every agent $i \in N$,
    \item $\mp_i$ is \textbf{indirect utility minimizing}: $\mp_i \in \argmin_{\mq_i \in \-R_{\geq 0}^m : \InAngles{\mx, \mq_i} \le B_i} v_i(\mq_i, B_i)$ for every agent $i \in N$,
    \item $\mx$ is \emph{profit-maximizing}: for every $j \in M$, $\sum_i p_{ij} \leq 1$; and whenever $x_j >0$ then $\sum_i p_{ij} = 1$.
\end{enumerate}
\end{definition}

\paragraph{Equivalence between \Cref{def::Lindahl-equilibrium} and \Cref{def:Lindahl-equivalent}} We note that the only difference between \Cref{def:Lindahl-equivalent} and \Cref{def::Lindahl-equilibrium} is condition $(ii)$. Recall that in the original definition, $(ii)$ $\mx$ maximizes the utility $u_i(\cdot)$ under price $\mp_i$ and budget constraint $B_i$ for each agent $i$:
\begin{align*}\label{cond::Lindahl-equilibrium}
\mx \in \argmax_{\my \in \-R_{\geq 0}^m : \InAngles{\my, \mp_i} \le B_i} u_i(\my).\numberthis
\end{align*}
While in the new \Cref{def:Lindahl-equivalent}, $(ii)$ requires that $\mp_i$ minimizes the indirect utility $v_i(\cdot, B_i)$ under ``price" $\mx$ and budget constraint $B_i$ for each agent $i$.  
\begin{align}\label{cond::Lindahl-dual}
    \mp_i \in \argmin_{\mq_i \in \-R_{\geq 0}^m : \InAngles{\mx, \mq_i} \le B_i} v_i(\mq_i, B_i).
\end{align}
Using the fundamental duality between utility and the indirect utility function (\Cref{thm::main-duality}), we show that \eqref{cond::Lindahl-equilibrium} and \eqref{cond::Lindahl-dual} are equivalent. We defer the proof to \Cref{sec::direct-indirect}
\begin{lemma}\label{lemma:equivalence between fisher (ii)}
    Let $u_i: \-R^m_{\ge 0} \rightarrow \-R_{\ge 0}$ satisfies \Cref{assumption:c-nd-ln-qc} and $B_i > 0$. Then a pair $(\mx, \{\mp_i\})$ satisfies \eqref{cond::Lindahl-equilibrium} if and only if it satisfies \eqref{cond::Lindahl-dual}.
\end{lemma}

Let us first see how this equivalence implies \Cref{thm::main-duality}. Recall \Cref{def::Fisher-market} of a Fisher market equilibrium in the dual market $\+F = \{N, M, \{\dual{u}_i\}, \{B_i\}\}$ and consider $\{\dual{\mx}_i := \mp_i\}$ and $\dual{\mp} := \mx$. We note $\{\mx, \{\mp_i\}\}$ satisfies \Cref{def:Lindahl-equivalent} $(i)$ and $(iii)$ if and only if $\{\{\dual{\mx}_i\}, \dual{\mp}\}$ satisfies \Cref{def::Fisher-market} $(i)$ and $(iii)$, since we exchange the role of allocations and prices in two markets. It then remains to show the equivalence of condition $(ii)$ in the two definitions. Recall the definition of a dual utility $\dual{u}_i(\cdot) = \frac{1}{v(\cdot,B_i)}$. Thus the new condition $(ii)$ in \Cref{def:Lindahl-equivalent} says ``$\mp_i$ maximizes the dual utility $\dual{u}(\cdot)$ under price $\mx$ and budget $B_i$ for every agent $i$''. Due to the definition $\{\dual{\mx}_i := \mp_i\}$ and $\dual{\mp} := \mx$, we have $\{\mx, \{\mp_i\}\}$ satisfies \Cref{def:Lindahl-equivalent} $(ii)$ if and only if $\{\{\dual{\mx}_i\}, \dual{\mp}\}$ satisfies \Cref{def::Fisher-market} $(ii)$. Combine the above, we conclude $(\mx, \{\mp_i\})$ is Lindahl equilibrium of $\+L= \{N, M, \{u_i\}, \{B_i\}\}$ if and only if $\{\{\dual{\mx}_i\}, \dual{\mp}\}$ is a Fisher market equilibrium of $\+F = \{N, M, \{\dual{u}_i\}, \{B_i\}\}$. This completes the proof.
\end{proof}

\paragraph{Shmyrev program for Lindahl Equilibrium with CES utilities} We remark that a generalized Shmyrev program for computing the Lindahl equilibrium with CES utilities follows directly from \Cref{thm::main-duality}. As established by \Cref{ex:CES and indirect utility} and \Cref{thm::main-duality}, the Lindahl equilibrium $(\mathbf{x}, \{\mathbf{p}_i\}_{i \in N})$ for agents with the CES utilities: $$u_i(\mx) = \begin{cases} \left( \sum_j a_{ij} x_{j}^{\rho_i} \right)^{1 / \rho_i} & \rho_i \in (- \infty, 1]\\ \min \left\{\frac{x_{ij}}{a_{ij}}\right\} & \rho_i = -\infty \end{cases}.$$ corresponds to a Fisher market equilibrium that uses dual CES utilities. 

Consequently, based on the Shmyrev program for Fisher markets with CES utilities~\cite{shmyrev2009algorithm, cheung2018dynamics}, the Lindahl equilibrium can be characterized by the following generalized Shmyrev program. In this program,  $b_{ij}$ denotes agent $i$'s spending on good $j$ with the constraints that $x_j = \sum_i b_{ij}$ (total spending on good $j$) and $\sum_j b_{ij} = B_i$ (agent $i$'s total budget):
\begin{align*}- \sum_{i: \rho_i \neq \{0, -\infty\}} &\frac{1}{\rho_i} \sum_j b_{ij} \log \frac{ b_{ij}}{a_{ij} x^{\rho_i}_j} - \sum_{i: \rho_i = -\infty} \sum_j b_{ij} \log \frac{a_{ij}}{ x_j} + \sum_{i : \rho_i = 0} \sum_j b_{ij} \log x_j, \\
&\text{ where } x_j = \sum_i b_{ij} \text{ and } \sum_j b_{ij} = B_i. \tag{Shmyrev-Lindahl-CES} \label{Shmyrev-Lindahl-CES}
\end{align*}
This program exhibits key properties based on the agents' utility parameters: (i) concave case ($\rho_i \geq 0$): When $\rho_i \geq 0$ for all agents, if we set $b_{ij} = \frac{a_{ij}}{\sum_{j'} a_{ij'}} B_i$ for any Cobb-Douglas agents, whose $\rho_i = 0$, the program becomes a concave function of the remaining agents' spending. The equilibrium points correspond to the maximum of this concave program. (ii) convex case ($\rho_i \leq 0$): Similarly, when $\rho_i \leq 0$ for all agents (and $b_{ij}$ is set as above for any Cobb-Douglas agents), the program becomes a convex function of the remaining agents' spending. The solution of this convex program corresponds to a Lindahl equilibrium allocation.

Additionally, we observe that \ref{Shmyrev-Lindahl-CES} simplifies to known results in two key scenarios with linear and Leontief utilities. First, for the linear case (where $\rho_i = 1$ for all $i$), the equation reduces to: $$\sum_i \sum_j \log b_{ij} \log \frac{b_{ij}}{a_{ij} x_{j}}.$$ This simplified form matches the potential function for the linear case given in \eqref{Shmyrev-Lindahl-Linear} \cite{kroer2025computing}. Second, for the Leontief case,  \ref{Shmyrev-Lindahl-CES} becomes:$$ - \sum_i \sum_j b_{ij} \log \frac{a_{ij}}{ x_{j}}.$$ Interestingly, we find that the potential function used in \citep{brandt2023balanced} to demonstrate the convergence of their best-response dynamic is precisely the above equation.

Furthermore, the proportional response dynamics correspond to the mirror descent procedure applied to the \eqref{Shmyrev-Lindahl-CES} program and have fast convergence rates. We elaborate on this interpretation and present the convergence results in \Cref{sec:prd:ces}.

\section{Application I: Homogeneous Fisher Market \& Lindahl Equilibrium} \label{sec:homogeneous:nsw}
We first apply our framework to the important class of  $1$-homogeneous utility functions. A utility function $u(\mx)$ is $1$-homogeneous if, for any vector of goods $\mx \geq 0$ and any scalar $\lambda \geq 0$, it satisfies the condition:$ u(\lambda \mx) = \lambda u(\mx)$. Notably, it includes the Constant Elasticity of Substitution (CES) utility function, which itself covers linear utilities and Leontief utilities as special cases.

\subsection{Duality of Homogeneous Utilities} A key property of our framework is that it preserves homogeneity, as formalized in the following lemma. We defer the proof to \Cref{sec:proof duality-homogeneous}.
\begin{lemma} \label{lem::duality-homogeneous}
If a utility function u is continuous, monotone, concave, and $1$-homogeneous, then its dual utility function, $\dual{u}$, is also continuous, monotone, concave, and $1$-homogeneous.
\end{lemma}

\begin{example}[Nested-CES Utility] \label{example:nested-ces}
Nested-CES utility functions \citep{keller1976nested} generalize the standard CES model by grouping goods into "nests." The intuition is hierarchical: utility is calculated in a bottom-up tree structure. In this tree, each leaf node corresponds to exactly one individual good. At each non-leaf node $\mathcal{I}$, a CES function aggregates the utilities of its children $\mathcal{C}(\mathcal{I})$, which can be individual goods (leaves) or other composite nests: 
$$u^{\mathcal{I}}(\mathbf{x}) = \left(\sum_{\mathcal{I}_c \in \mathcal{C}(\mathcal{I})} a^{\mathcal{I}_c} \cdot \left(u^{\mathcal{I}_c}(\mathbf{x})\right)^{\rho(\mathcal{I})} \right)^{1 / \rho(\mathcal{I})}, $$
where $a^{\mathcal{I}_c}$ are preference parameters and $\rho(\mathcal{I}) \in (-\infty, 1]$ is the substitution parameter for that nest.

For instance, $ u(\mathbf{x}) = \left(6\cdot\left[\left(3\cdot x_1^{0.2} + x_2^{0.2}\right)^{1/0.2}\right]^{0.7} + 8\cdot x_3^{0.7} \right)^{1/ 0.7}. $ represents a two-level tree. First, $x_1$ and $x_2$ are aggregated into a composite good, which is then aggregated with $x_3$.
\paragraph{Dual utility function of nested-CES utility}
For nested-CES utilities, a key property is that the dual utility function, $\dual{u}$, retains this nested-CES structure.  The dual function for a  non-leaf node $\mathcal{I}$ is:
$$ \dual{u}^{\mathcal{I}}(\mathbf{x}) = \left(\sum_{\mathcal{I}_c \in \mathcal{C}(\mathcal{I})} \left(a^{\mathcal{I}_c}\right)^{1 - \dual{\rho}(\mathcal{I})} \cdot \left(u^{\mathcal{I}_c}(\mathbf{x})\right)^{\dual{\rho}(\mathcal{I})} \right)^{1 / \dual{\rho}(\mathcal{I})}, $$
where the dual elasticity parameter is $\dual{\rho}(\mathcal{I}) = \frac{\rho(\mathcal{I})}{\rho(\mathcal{I}) - 1}$. The final dual utility $\dual{u}(\mathbf{x})$ is a normalization of the root node's function, $\frac{1}{B_i} \dual{u}^{\text{root}}(\mathbf{x})$.
\end{example}

\subsection{Lindahl Equilibrium and Nash Social Welfare maximization} It is a cornerstone result that for Fisher markets with $1$-homogeneous utilities, the equilibrium allocation is characterized by the solution to the Eisenberg-Gale convex program, which maximizes the weighted Nash Social Welfare (NSW):$$\max_{\{\mx_i \geq 0\}_i} \prod_{i} u_i(\mx_i)^{B_i}\text{ subject to } \sum_{i} \mx_i \leq \mathbf{1},$$ where $B_i$ is the budget of agent $i$, $\mx_i$ is the allocation vector for agent $i$, and $\mathbf{1}$ is an $m$-dimensional vector of ones representing the total supply of each good.

A parallel NSW formulation exists for the public goods setting, reflecting shared consumption constrained by the total societal budget:

\begin{align*}\max_{\mx \geq 0} \prod_{i} u_i(\mx)^{B_i}\text{ subject to } \|\mx\|_1 \leq \sum_{i} B_i.  \numberthis \label{nsw:lindahl-eq}
\end{align*}

Previous research has demonstrated that the solution to \eqref{nsw:lindahl-eq} coincides with the Lindahl equilibrium under certain restrictive settings:
\begin{enumerate}
    \item It has been shown that for agents with separable, homogeneous utilities of positive degree, i.e.,  CES utility functions with $\rho \in [0, 1]$, the Lindahl equilibrium allocation can be computed by \eqref{nsw:lindahl-eq} \cite{fain2016core}.
    \item For Leontief utility functions, the allocation that maximizes the NSW has also been shown to coincide with the Lindahl equilibrium\cite{brandt2023balanced}.
\end{enumerate}

Our work generalizes this connection by showing that NSW maximization solution coincides with the Lindahl equilibrium for general 1-homogeneous, concave utilities.
\begin{theorem}[Lindahl Equilibrium] \label{thm:intro-NSW-Lindahl}
If the utility function $u_i$ for each agent $i$ is continuous, monotone, concave, and $1$-homogeneous, then the Lindahl equilibrium allocation $\mx$ is characterized by the solution to the NSW maximization problem \eqref{nsw:lindahl-eq}
\end{theorem}

The proof of Theorem \ref{thm:intro-NSW-Lindahl} relies on the established duality between Lindahl and Fisher markets. A Lindahl equilibrium where agents have $1$-homogeneous utility functions, $u_i$, can be mapped to an equivalent Fisher market equilibrium. In this dual Fisher market, the agents are endowed with the corresponding dual utility functions, $\dual{u}_i$. Per Lemma~\ref{lem::duality-homogeneous}, these dual utility functions are also $1$-homogeneous. With this preservation of homogeneity, the dual Fisher market equilibrium allocation is characterized by the solution to its own weighted NSW maximization problem, and the dual of this Fisher NSW problem is $$\min_{\mp \geq 0} \sum_j p_j + \sum_{i} B_i \log \dual{v}_i(\mp, B_i),$$ where $\dual{v}_i(\mp, B_i)$ is the indirect utility function of $\dual{u}_i$, and the optimal dual variable $\mp$ corresponds to the Fisher market equilibrium price. Additionally, the equilibrium price condition $\sum_j p^*_j = \sum_i B_i$ further simplifies the program to $$\min_{\mp \geq 0}  \sum_{i} B_i \log \dual{v}_i(\mp, B_i) \text{ subject to } \sum_j p_j = \sum_i B_i.$$ The final step is to translate this result back to the original Lindahl context. Using the dual relationship $u_i(\mathbf{p}) = 1 / \dual{v}_i(\mathbf{p}, B_i)$ and reinterpreting the Fisher price vector $\mathbf{p}$ as the public good allocation $\mathbf{x}$, we recover the log-transformed Lindahl NSW maximization problem \eqref{nsw:lindahl-eq}.

\subsection{NSW with General Concave Utilities}

The exact equivalence between equilibrium and the NSW optimum is specific to the $1$-homogeneous case. For general concave utilities, the two concepts diverge, as illustrated by the following example.

\begin{example}[Lindahl equilibrium and NSW divergence]
Consider a Lindahl economy with two agents and two public goods, $x_1$ and $x_2$.  Let both agents have an initial budget of $B_1 = B_2 = 0.5$. Agent 1's utility function is $u_1(\mx) = 2 x_1 + x_2$, and agent 2's utility function is $u_2(\mx) = \log(x_1 + 2 x_2)$.
In this economy, the Lindahl equilibrium allocation is the equal split $\mx^{\+L} = (0.5, 0.5)$.

In contrast, the NSW-optimal allocation, which maximizes the product $u_1(\mx) \cdot u_2(\mx)$, is found to be: $\mx^{\text{NSW}} = \left(2 - \frac{3}{W(3e)}, \frac{3}{W(3e)} - 1\right) \approx (0.14545, 0.85455)$
where $W(\cdot)$ denotes the Lambert W function. This demonstrates that the Lindahl equilibrium and the NSW-optimal allocation diverge when not all agents possess $1$-homogeneous utilities.
\end{example}

Given this divergence, we analyze the approximation quality of the Lindahl equilibrium. For Fisher markets, it is known that any equilibrium achieves an approximation ratio of $(1/e)^{1/e}$ for the optimal NSW \cite{garg2025approximating}. We extend this result to Lindahl equilibria in public goods markets.

\begin{theorem}\label{theorem::non-homogenous}
Given a Lindahl equilibrium instance with concave utility functions such that $u_i(\mathbf{0}) = 0$, consider any Lindahl equilibrium $(\mx, \{\mp_i\})$ and a Nash welfare maximizing allocation $\my$ such that $\langle \my, \mathbf{1} \rangle \leq \sum_i B_i$. Then, the budget-weighted geometric mean of utilities satisfies:
\[
\left(\prod_{i} u_i(\mx)^{B_i}\right)^{\frac{1}{\sum_{i} B_i}} \geq \left(\frac{1}{e}\right)^{\frac{1}{e}} \left(\prod_{i} u_i(\my)^{B_i}\right)^{\frac{1}{\sum_{i} B_i}}.
\]
\end{theorem}
The proof of the theorem is provided in \Cref{sec:non-homogenous-proof}. Furthermore, we demonstrate that the approximation bound of $(1/e)^{1/e}$ is tight (see \Cref{example::etoe}).

\section{Application II: Market Dynamics} \label{sec::marketdynamics}
Market dynamics has a long history in the general equilibrium theory. In Fisher market, two main dynamic procedures are widely considered: proportional response dynamics and t\^atonnement.

\subsection{Proportional Response Dynamics} \label{sec:prd}

Among various distributed market algorithms, proportional response dynamics (PRD) has garnered significant attention~\citep{WZ2007,LLSB08,Zhang2011,birnbaum2011distributed,cheung2018dynamics,BranzeiMehtaNisan2018NeurIPS,CheungHN19,gao2020first,BranzeiDevanurRabani2021,Branzei2021-Sigecom,CheungLP2021,ZhuCX2023WWW,kolumbus2023asynchronous,li2024proportional,cheng2024tight}, partly due to its simple implementation in networked markets. 
While we have discussed convergence of PRD for CES utilities, a key line of inquiry has focused on the convergence of PRD in Fisher markets where agents have utilities satisfying the gross substitutes and normal goods property. 

To formally define these properties, we first recall the Marshallian demand. Given a utility function $u: \-R_{\geq 0}^m \rightarrow \-R_{\geq 0}$ with budget $B > 0$, the Marshallian demand function under price $\mp \in \-R^m_{\geq 0}$ is $$\demand(\mp, B) := \argmax_{\mx \in \-R_{\geq 0}^m: \InAngles{\mp, \mx} \le B} u(\mx).$$

A utility function satisfies the gross substitutes property if an increase in the price of one good does not decrease the demand for any other good. 
\begin{definition}[Gross Substitutes]
    A utility function $u$ satisfies the Gross Substitutes (GS) property if for any price vectors $\mp$ and $\mp'$ such that $\mp \le \mp'$ and for any good $j$ with $p_j = p'_j$:
    \begin{align*}
        \demandnb_j(\mp, B) \le \demandnb_j(\mp',B).
    \end{align*}
\end{definition}

Additionally, normal good property states that the demand for any good is non-decreasing as the agent's budget increases.
\begin{definition}[Normal Goods]
    A utility function $u$ satisfies the Normal Goods property if for any price vector $\mp$ and any budget $0 < B < B'$:
    \begin{align*}
        \demandnb_j(\mp, B) \le \demandnb_j(\mp,B').
    \end{align*}
\end{definition}

A recent result~\citep{cheung2025proportional} shows that if all agents have utilities satisfying both the gross substitutes and normal goods properties, the PRD converges to the market equilibrium. We define $b_{ij}^t$ as the spending of agent $i$ on good $j$ at round $t$. This spending is the core variable updated by the dynamics. The general form of the PRD update rule in this case is as follows:
\begin{align}b_{ij}^{t+1} = B_i \frac{x^t_{ij} \nabla_j u_i(\mx^t_i)}{\sum_{j'} x^t_{ij'} \nabla_{j'} u_i(\mx^t_i) }  \tag{PRD-Fisher-GS}\label{prd:updating-rule}\end{align}
where $B_i$ is the budget of agent $i$, $x_{ij}^t =b_{ij}^t / p_j^t$ is the allocation of good $j$ at time $t$ (based on spendings $b_{ij}^t$ and prices $p_j^t = \sum_i b_{ij}^t$), and $\nabla_j u_i(\mx^t_i)$ is the marginal utility of good j given the allocation vector $\mx^t_i$. 

While this dynamic is known to converge, the intuition behind \ref{prd:updating-rule} is not immediately apparent, making the underlying agent behavior difficult to interpret.

\paragraph{A Dual Interpretation via Roy's Identity} 
We provide a new interpretation of the updating rule, \ref{prd:updating-rule}, by leveraging our duality framework. This framework establishes a correspondence between a Lindahl equilibrium and a Fisher market equilibrium and also enables the transformation of the Fisher market's PRD to the Lindahl equilibrium.

Consider a primal Lindahl equilibrium with agents having primal utility functions $\{u_i\}$. This maps to a corresponding Fisher market with agents having dual utility functions $\{\dual{u}_i\}$. For convenience, we use $\dual{\mx}$ and $\dual{\mp}$ to denote the Fisher market's allocation and prices, respectively.

Let us apply the PRD update rule \ref{prd:updating-rule} within this dual Fisher market. The rule for updating agent $i$'s spending on good $j$ at iteration $t+1$ is: \begin{align*}b_{ij}^{t+1} = B_i \frac{\dual{x}^t_{ij} \nabla_j \dual{u}_i(\dual{\mx}^t_i)}{\sum_{j'} \dual{x}^t_{ij'} \nabla_{j'} \dual{u}_i(\dual{\mx}^t_i) } = B_i \frac{p^t_{ij} \nabla_{p^t_{ij}} v_i(\mp^t_i, B_i)}{\sum_{j'} p^t_{ij'} \nabla_{p^t_{ij'}} v_i(\mp^t_i, B_i) }. \end{align*}
The second equality holds due to our duality framework, the correspondence between the components of the Fisher market and the Lindahl equilibrium. Specifically, the Fisher allocation $\dual{\mx}^t$ corresponds to the Lindahl price vector $\mathbf{p}^t$. Furthermore, the gradients of the dual utility $\nabla \dual{u}_i$ are directly related to the gradients of the standard indirect utility function, as $\dual{u}_i(\mathbf{p}_i) = \frac{1}{v_i(\mathbf{p}_i, B_i)}$.

The RHS can be simplified further. The indirect utility function $v_i(\mathbf{p}_i, B_i)$ is $0$-homogeneous in prices and budget; that is $v_i(\lambda \mathbf{p}_i, \lambda B_i) = v_i(\mathbf{p}_i, B_i)$ for any $\lambda > 0$. By Euler's theorem for homogeneous functions, the denominator sum is $\sum_{j'} p^t_{ij'} \nabla_{p^t_{ij'}} v_i(\mathbf{p}^t_i, B_i) = -B_i \nabla_{B_i} v_i(\mathbf{p}^t_i, B_i)$.

Substituting this result back into the equation yields:$$b_{ij}^{t+1} = B_i \frac{p^t_{ij} \nabla_{p^t_{ij}} v_i(\mathbf{p}^t_i, B_i)}{-B_i \nabla_{B_i} v_i(\mathbf{p}^t_i, B_i) } = p^t_{ij} \left[ - \frac{\nabla_{p^t_{ij}} v_i(\mathbf{p}^t_i, B_i)}{\nabla_{B_i} v_i(\mathbf{p}^t_i, B_i)} \right].$$ The term in the brackets is precisely Roy's identity, which is equal to the agent's Marshallian demand for good $j$, denoted $\demandnb_{ij}(\mp_i^t,B_i)$. Therefore, the update rule in the dual market, when viewed from the primal perspective, becomes:

$$b_{ij}^{t+1} = p^t_{ij} \cdot \demandnb_{ij}(\mathbf{p}_i^t, B_i)$$
This derivation reveals that the opaque PRD rule \eqref{prd:updating-rule} is, in the dual setting, equivalent to a simple and intuitive expenditure best-response dynamic in the primal Lindahl market.

\paragraph{Total Complements: The Dual of Gross Substitutes and Normal Goods}

Given this duality, our objective is to characterize the conditions on the primal utility functions $\{u_i\}$ that guarantee convergence of this new Lindahl dynamic.

Prior work establishes that convergence of the dual Fisher market dynamic, \eqref{prd:updating-rule}, is guaranteed if the dual utilities $\{\dual{u}_i\}$ satisfy the gross substitutes and normal goods properties.\footnote{\label{fn:regcond}Two technical regularity assumptions on the dual demand $\dualdemand_{i}(\dual{\mathbf{p}}, B_i)$ are also required \citep{cheung2025proportional}: (i) surjectivity (any feasible dual allocation $\dual{\mx}_i \neq 0$ can be generated by some price vector $\dual{\mathbf{p}} > 0$) and (ii) uniqueness (demand is unique for any $\dual{\mathbf{p}} > 0$).} These two properties are jointly captured by the condition presented in \Cref{lem::GS-NG}, whose proof is deferred to \Cref{sec::GS-NG-lemma-proof}.

\begin{lemma}\label{lem::GS-NG}
    If a utility function $u$ satisfies the gross substitutes and normal goods properties, then for any price vectors $\mp > 0$ and $\mp' > 0$, if good $j$ satisfies $\frac{p_j}{p'_j} \geq \max\left\{ \max_k\left\{\frac{p_k}{p'_k}\right\}, 1 \right\}$ then the demand for good $j$ at price $\mp$ will be no greater than its demand at price $\mp'$: $\demandnb_j(\mp, B) \leq \demandnb_j(\mp', B)$.
\end{lemma}

These requirements on the dual utilities (namely, gross substitutes and normal goods, or equivalently, \Cref{lem::GS-NG}) imply a corresponding, ``dual'' structure on the primal utilities $\{u_i\}$, which we define as the \textbf{total complements} property.

The name ``complements'' highlights its inverse relationship to the gross substitutes. Informally, the total complements property relates allocation changes to supporting price changes. It requires that for a good $j$ to experience \emph{the largest proportional decrease} in demand (comparing a new allocation $\mathbf{x}'$ to an old one $\mathbf{x}$), its new supporting price $p'_j$ must have \emph{increased} relative to the old price $p_j$ (or at least not decreased).

\begin{definition}[Total Complements]\label{def:tc}
A utility function $u$ satisfies the \textbf{total complements } property if $\demand(\mathbf{p}, B_i) > 0$ for any price vector $\mathbf{p}$, and for any allocation $\mx = \demand(\mathbf{p}, B_i)$ and $\mx' = \demand(\mathbf{p}', B_i)$, if good $j$ satisfies:
\[ \frac{x_j}{x'_j} \geq \max\left\{\max_k \left\{\frac{x_k}{x'_k}\right\}, 1\right\}, \]
then it must be that $p_j \leq p_j'$.
\end{definition}

\begin{example}
CES utility functions, $u_i(\mx) = \left(\sum_j a_{ij} x_j^{\rho_i}\right)^{1 / \rho_i}$, satisfy the total complements property if the elasticity parameter $\rho_i \in (-\infty,0]$.
\end{example}
This definition implies that for any given demand allocation $\mathbf{x}$, the supporting price vector $\mathbf{p}$ must be unique, mirroring the regularity conditions required in the dual market.\footnote{See footnote~\ref{fn:regcond}.} This technically excludes Leontief utilities, as a Leontief allocation can be supported by a range of prices. However, our convergence results extend to Leontief utilities, as they are dual to linear utilities.

\paragraph{PRD in Lindahl Equilibrium with Total Complements} We now state PRD for the Lindahl equilibrium where agents have total complements utilities. 

The dynamic updates agent $i$'s spending contribution to good $j$ as:$$b_{ij}^{t+1} = p_{ij}^t \cdot \demandnb_{ij}(\mathbf{p}_i^t, B_i)$$where $\demandnb_{ij}(\mathbf{p}_i^t, B_i)$ is agent $i$'s Marshallian demand given their personalized prices $\mathbf{p}_i^t$ and total budget $B_i$. This update rule has a clear economic interpretation as an expenditure best-response. At each iteration $t$, agent $i$ observes their personalized prices $\mathbf{p}_i^t$ and sets their next-period budget $b_{ij}^{t+1}$ to their current optimal expenditure. 

The personalized Lindahl price $p_{ij}^t$ is determined by the agents' contributions. It represents agent $i$'s proportional share of the total spending on good $j$, $x_j^t = \sum_k b_{kj}^t$: $$p_{ij}^t = \frac{b_{ij}^t}{ x_j^t} = \frac{b_{ij}^t}{\sum_k b_{kj}^t}.$$

The following theorem establishes the convergence of this dynamic. The proof is provided in \Cref{sec::proof:theorem:PRD:lindahl:tc}.
\begin{theorem}[PRD for Lindahl Equilibrium with Total Complements]\label{theorem:PRD:lindahl:tc}If the utility function for each agent is strictly increasing, strictly concave, and satisfies the total complements condition, then the PRD converges to the Lindahl equilibrium. Moreover, the empirical convergence rate of $\mx^t$ to the optimal solution is $O(1/T)$. \end{theorem}

Additionally, the PRD dynamic also converges for Leontief utilities, as it is equivalent to the PRD for a Fisher market with linear utilities.

\paragraph{PRD for Lindahl Equilibrium with Gross Substitutes} The total complements framework provides a convergent, interpretable dynamic, but it excludes the important class of utilities satisfying the standard gross substitutes property (e.g., CES functions with $\rho > 0$).

To provide a convergent dynamic for this class, we adapt the \emph{original} PRD spending update rule, \ref{prd:updating-rule}, from the Fisher market setting. The dynamic for agent $i$'s spending contribution to good $j$ is:\begin{align} \tag{PRD-Lindahl-GS}\label{prd:lindahl-gs}b_{ij}^{t+1} = B_i \frac{x^t_{j} \nabla_j u_i(\mx^t)}{\sum_{j'} x^t_{j'} \nabla_{j'} u_i(\mx^t) }\quad \text{where } p^t_{ij} = \frac{b_{ij}^t}{x_j^t} \text{ and } x_j^t = \sum_i b_{ij}^t.\end{align}
We establish the convergence of this dynamic under standard gross substitutes conditions. The proof is provided in \Cref{sec::proof:prd-lindahl-gs}.

\begin{theorem}[PRD Convergence for Lindahl with Gross Substitutes]\label{thm:prd-lindahl-gs}Suppose all agent utility functions $u_i$ are strictly increasing, strictly concave, and satisfy the gross substitutes and normal goods properties. Then the PRD, \ref{prd:lindahl-gs}, converges to the Lindahl equilibrium.\end{theorem}
We note that for CES utilities with $\rho > 0$, this dynamics converges at a rate of at least $O(1/T)$, as it can be interpreted as mirror descent on the generalized Shmyrev program,   \eqref{Shmyrev-Lindahl-CES}. Further details on this are provided in \Cref{sec:prd:ces}.

We leave the convergence rate for the general case as an open question.

\paragraph{PRD for Fisher Market with Total Complements}

To complete our dual analysis, we map the PRD for the Lindahl equilibrium with gross substitutes utilities, \ref{prd:lindahl-gs}, back to its corresponding Fisher market. This procedure yields a PRD for a Fisher market populated by agents with total complements utilities:
\begin{align}\tag{PRD-Fisher-TC} \label{prd:fisher-tc}b_{ij}^{t+1} = p_{j}^t \cdot \demandnb_{ij}(\mathbf{p}^t, B_i) \quad \text{where } p^t_{j} = \sum_i b_{ij}^t.\end{align}

This update rule, \ref{prd:fisher-tc}, coincides with the known PRD for Fisher markets with CES utilities where $\rho_i \in (-\infty, 0]$ \citep{cheung2018dynamics}, which aligns perfectly with our definition of the total complements class.

\begin{theorem} \label{thm:fisher:tc}If the utility function for each agent is strictly increasing, strictly concave, and satisfies the total complements property, then PRD converges to the Fisher market equilibrium.\end{theorem}
The proof is omitted, as it is identical to that of \Cref{thm:prd-lindahl-gs} with price and allocation exchanged.

We note that this dynamic can also be interpreted as a t\^atonnement process. Since $p_j^{t+1} = \sum_i b_{ij}^{t+1}$, the price update rule becomes:$$p_j^{t+1} = \sum_i p_j^t \cdot \demandnb_{ij}(\mathbf{p}^t, B_i) = p_j^t \sum_i \demandnb_{ij}(\mathbf{p}^t, B_i) = p_j^t (1 + z_j^t),$$where $z_j^t$ is the excess demand for good $j$. 

We omit the proof here, as the proof is identical to the proof of \Cref{thm:prd-lindahl-gs} by exchanging the price and the allocation.

\subsection{T\^atonnement Dynamics} Finally, we apply our dual framework to t\^atonnement dynamics. \label{sec::tatonnement}

\paragraph{T\^atonnement in Fisher markets} T\^atonnement, a concept introduced in \cite{walras1896etudes}, describes an iterative process for finding equilibrium in a private good market. 

In a Fisher market, for markets where consumers have nested-CES utility functions (see Example~\ref{example:nested-ces}) excluding the case of having a linear component, the following price-adjustment process is guaranteed \cite{cheung2014analyzing} to converge to the equilibrium when the step size $\Gamma_j$ is big enough:
\begin{align*}p_j^{t+1} = p_j^t \cdot \text{exp}\left(\frac{\min\{z_j^t, 1\}}{\Gamma_j}\right) \text{ for all good $j$}. \numberthis \label{dynamics:tat:fisher}
\end{align*}
where $\Gamma_j$ is the step size and $\mz$ is the excess demand vecotr: $\mz^t \triangleq \sum_i \demand_{i}(\mp^t) - 1$.

To state the precise convergence condition, we define indices based on the hierarchical structure of the nested-CES utility $u$. Let $\+P$ be the set of all root-to-leaf paths in the utility tree. \begin{enumerate}
    \item[-] \emph{Substitutes Index} ($\+S\+I$): This index is the minimum of the summed $\min\{\rho(\+I)/(\rho(\+I) - 1), 0\}$ values along any root-to-leaf path:$$  \+S\+I(u) = \min_{P \in \+P} \sum_{\+I \in P} \min\{\rho(\+I)/(\rho(\+I) - 1), 0\}.$$
    \item[-] \emph{Complements Index} ($\+C\+I$): Similarly, the Complements Index is defined using the primal elasticity parameters $\rho$:$$  \+C\+I(u) = \min_{P \in \+P} \sum_{\+I \in P} \min\{\rho(\+I), 0\}.$$A known convergence guarantee for the Fisher market dynamic relies on the Substitutes Index.
\end{enumerate}
\begin{theorem}[\cite{cheung2014analyzing}]
    Consider a Fisher market where agents' utility functions are nested CES. The procedure \eqref{dynamics:tat:fisher} is guaranteed to converge to the equilibrium if the step size $\Gamma_j$ satisfies $\Gamma_j \geq \left[8 - \frac{252}{25} \min_i \+S\+I(u_i)\right]$.
\end{theorem}

\paragraph{A Dual T\^atonnemnt for Lindahl Equilibrium} By applying our dual framework, we can directly translate the t\^atonnement dynamics from the Fisher market into a dual process for a corresponding Lindahl equilibrium, with nested-CES utilities that do not contain Leontief components: 
\begin{align*}b_{ij}^{t+1} = B_i \frac{x^t_{j} \nabla_j u_i(\mx^t)}{\sum_{j'} x^t_{j'} \nabla_j u_i(\mx^t) } \text{ and }x_j^{t+1} = x_j^t \cdot \text{exp}\left(\frac{\min\{o_j^t, 1\}}{\Gamma_j}\right)\text{ for all good $j$}. \numberthis \label{dynamics:tat:Lindahl}
\end{align*}
where $\Gamma_j$ is the step size and $\mo$ is the overpayment: $\mo^t \triangleq  \frac{\sum_i b_{ij}^t}{x^t_j} - 1$.

For nested-CES utilities, the spending update rule simplifies into a multiplicative form. The update rule for the spending $b_{ij}^{t+1}$ is determined by the agent's total budget $B_i$ multiplied by a product of allocation ratios. This product is computed over all nodes $\+I$ along the path from the root of the utility tree to the leaf corresponding to good $j$. The specific form is
$$b_{ij}^{t+1} = B_i \frac{x^t_{j} \nabla_j u_i(\mx^t)}{\sum_{j'} x^t_{j'} \nabla_j u_i(\mx^t) } = B_i \prod_{\+I: \text{root} \rightarrow x_j} \frac{a^{\+I_j} \left(u^{\+I_j}(\mx)\right)^{\rho(\+I)}}{\sum_{\+I_c} a^{\+I_c} \left(u^{\+I_c}(\mx)\right)^{\rho(\+I)}}$$ where the product $\prod_{\+I: \text{root} \rightarrow x_j}$ is taken over each node $\+I$ on the path from the root to the leaf $x_j$, $\+I_j$ denotes the child node of $\+I$ that lies on this path, and $\+I_c$ iterates over all children of the node $\+I$.

\begin{theorem}\label{thm:tat:lindahl}
    Consider a Lindahl equilibrium where agents' utility functions are nested CES. The procedure \eqref{dynamics:tat:Lindahl} is guaranteed to converge to the equilibrium if the step size $\Gamma_j$ satisfies $\Gamma_j \geq \left[8 - \frac{252}{25} \min_i \+C\+I(u_i)\right]$.
\end{theorem}

\section{Application III: Markets with Chores}
\label{sec::chores}
This section analyzes the allocation of chores, which are items that incur disutility to agents. We examine two distinct market frameworks: the Fisher market for private chores and the Lindahl equilibrium for public chores.

\subsection{Fisher Market with Private Chores}

A Fisher market instance for private chores, $\mathcal{F} = (N, M, \{d_i\}, \{B_i\})$, consists of a set of $m$ divisible chores, $M$ and a set of $n$ agents, $N$. We assume a unit supply of each chore. Each agent $i \in N$ has a disutility function $d_i: \mathbb{R}_{\geq 0}^m \to \mathbb{R}_{\geq 0}$, which is assumed to be non-decreasing and not identically zero.

In this market, each agent $i$ is compensated with a payment for the chores they perform and is subject to an earning constraint $B_i > 0$. The goal is to find a Fisher market equilibrium for chores, also known as a \emph{competitive equilibrium (CE) for chores}. This equilibrium is a pair of chore prices and allocations, $(\mathbf{p}, \{\mathbf{x}_i\})$, such that: (i)  Every agent receives their disutility-minimizing bundle of chores, subject to achieving their exact earning constraint; and (ii) Every chore is fully allocated, i.e., the market clears.

\begin{definition}[Fisher Market Equilibrium for Chores] \label{definition:Fisher-chores}
Given a Fisher market instance $\mathcal{F}=(N, M, \{d_i\}, \{B_i\})$ for chores, the allocations and prices $((\mathbf{x}_i), \mathbf{p})$ form a \emph{Fisher market equilibrium / competitive equilibrium} if the following hold:
\begin{enumerate}
    \item Disutility minimizing subject to earning constraint: $\mathbf{x}_i \in \argmin_{\my_i \in \mathbb{R}_{\geq 0}^m: \langle \mp, \my_i \rangle = B_i}d_i(\my_i)$ for every agent $i \in N$;
    \item Market Clearing: For every chore $j \in M$, $\sum_{i \in N} x_{ij} = 1$.
\end{enumerate}
\end{definition}

This framework is motivated by practical fair chore division problems, such as teachers dividing teaching loads or roommates dividing household chores. A prominent special case is the \emph{competitive equilibrium for equal income (CEEI)}, where $B_i = 1$ for all agents. A CEEI allocation is known to satisfy several desirable efficiency and fairness properties~\citep{bogomolnaia2017competitive}: (i) \emph{Pareto-Optimal}: No other feasible allocation exists that makes at least one agent better off (less disutility) without making any other agent worse off;
 (ii) \emph{Envy-Free}: No agent prefers another agent's bundle of chores, i.e., $d_i(\mathbf{x}_i) \le d_i(\mathbf{x}_j)$ for all $i, j \in N$; (iii) \emph{Guarantees Fair Share}: Each agent's disutility is no more than their disutility from performing an equal split of all chores, i.e., $d_i(\mathbf{x}_i) \le d_i(\frac{1}{n} \cdot \mathbf{1})$; (iv) \emph{In the Weak Core}: It is stable against "blocking coalitions", i.e, there is no subset of agents $\+A \subseteq N$ and an allocation $\{\my_i\}_{i \in \+A}$ such that $\sum_{i \in \+A} \my_i \ge \frac{|\+A|}{n} \cdot \mathbf{1}$ and every agent $i \in \+A$ strictly prefers $\my_i$ over $\mx_i$.

The seminal work by \cite{bogomolnaia2017competitive} characterizes the set of CE when the disutility functions are convex and 1-homogeneous. They propose the EG-type program that minimizes the product of agents' disutilities and show that every KKT point of the program corresponds to a competitive equilibrium, with the exception of the KKT points that assign zero disutility to some agent. To avoid the zero-disutility issue, one can optimize the logarithm of the product of disutilities. While this avoids zero points, this new non-convex program introduces poles—infeasible points on the boundary of the feasible region that attract iterative methods and drive the objective to negative infinity~\citep{boodaghians2022polynomial,chaudhury2022competitive}. While specialized iterative methods have been developed to find non-zero KKT points \citep{boodaghians2022polynomial,chaudhury2022competitive}, these methods are sophisticated as they require solving non-linear programs at each iteration. These methods are not practical for large-scale problems as shown in \citep{chaudhury2024competitive}, who ask the following question:
\begin{equation*}
    \textit{Does there exist an optimization formulation of the chores problem that avoids the poles issue?}
\end{equation*}
\citet{chaudhury2024competitive} resolved the question for the specific case of linear disutilities by proposing a new program with no poles, which opens up the CE for chores problem to more standard iterative methods. However, constructing a program that avoids the poles issue for general convex and 1-homogeneous disutilities remained open. 

We address the gap and propose a program \ref{Fisher-Chores} that avoids the poles issue for general convex and $1$-homogeneous disutilities, whose KKT points correspond to CE (\Cref{theorem: fisher-chores-KKT}). Our program is constructed by leveraging the concept of \emph{indirect disutility functions}. Specifically, we give a Roy's identity-like characterization of optimal allocation under a certain price using the indirect disutility function. 

In the following, we first review the linear case and discuss the challenges in the general setting. Then we formally introduce the concept of the indirect disutility function and its properties. In the end, we give our program. 

\paragraph{Technical Overview} We briefly review \citep{chaudhury2024competitive}'s program for linear disutilities and discuss why it is not obvious to extend it to the more general setting. Here we assume each agent $i$' has linear disutility $d_i(\mx_i) = \sum_j d_{ij} x_{ij}$ where $d_{ij} \ge 0$. The EG-style program for CE is \ref{EG-Chores-Linear}, which looks similar to the EG program in the goods setting. However, unlike the goods setting, \ref{EG-Chores-Linear} is a non-convex program, and the objective function can tend to negative infinity (poles) within the feasible region.

\begin{figure}[h]
  \centering

  \noindent\rule{\linewidth}{0.4pt}\vspace{0.6em}

  \begin{tabularx}{\linewidth}{@{}C C@{}}
  $\displaystyle
    \begin{aligned}
        \inf_{\mx \ge 0} \,& \sum_{i \in N} B_i \log \InParentheses{\sum_{j \in M} d_{ij} x_{ij}}\\
        \text{s.t.}\,& \sum_{i \in N} x_{ij} = 1, \quad  \text{for all } j \in M.
    \end{aligned}
  $
  &
  $\displaystyle
    \begin{aligned}
        \sup_{{\bf \beta} \ge 0, \mp \ge 0} \,& \sum_{j\in M} p_j
        - \sum_{i \in N} B_i \log\beta_i \\
        \text{s.t.}\,& p_j\le \beta_id_{ij}, \text{ for all } i\in N, j \in M \\
        \,& \sum_{j \in M} p_j =\sum_{i \in N} B_i 
    \end{aligned}
  $
  \\
  \small (\setword{EG-Chores-Linear}{EG-Chores-Linear}) & \small (\setword{EG-Chores-Linear-Dual}{EG-Chores-Linear-Dual})
  \end{tabularx}

  \vspace{0.4em}\noindent\rule{\linewidth}{0.4pt} 
\end{figure}
Although duality does not hold for non-convex programs, \citet{chaudhury2024competitive} guess a ``dual" of \ref{EG-Chores-Linear} according to the dual program in the goods setting and propose \ref{EG-Chores-Linear-Dual}. This program avoids the poles issue since the constraints induce upper bounds for the prices and lower bounds for $\{\beta_i\}$ and thus gives an upper bound on the objective. They then show that any KKT point $(\mp, \beta)$ of this program corresponds to a CE. Specifically, given the price vector $\mp$, the corresponding allocations $\{\mx_i\}$ are obtained as follows: the allocation $x_{ij}$ is the dual variable (KKT multiplier) of the constraint $p_j \le \beta_i d_{ij}$. Thus these $N \times M$ constraints gives the allocations $\{x_{ij}\}_{i\in N,j\in M}$.

Analogy to the goods setting, we generalize the program to general utilities using the concept of indirect disutility functions and dual utilities in \ref{Fisher-Chores}. We note that, however, the fact that KKT multipliers correspond to the allocations is very special to the linear setting and does not hold in general. For the new program with general disutilities, when we get a price vector $\mp$ from a KKT point, it is not obvious which allocations $\{\mx_i\}$ make $(\mp, \{\mx_i\})$ a CE.  In fact, we only have $N$ constraints on the indirect disutility functions, one for each agent $i$, whose dual variables clearly do not correspond to allocations. We address this by establishing a Roy's Identity-type characterization of the optimal allocation given prices using subgradients of the indirect disutility function, which relies on the duality between direct and indirect disutility functions similar in the goods setting.

\paragraph{Indirect Disutility Function} 
We first introduce indirect disutility functions. Given a disutility function $d: \mathbb{R}_{\geq 0}^m \to \mathbb{R}_{\geq 0}$ that is non-decreasing and not identically $0$, its \emph{indirect disutility function} $h: \-R^m_{\ge 0} \times \-R_{>0} \rightarrow \-R_{\ge 0}$ is defined as follows: $h(\mp, B)$ is the minimum disutility required to achieve a total earning of $B > 0$ given a price vector $\mathbf{p} \ge \boldsymbol{0}$ such that $\mp \ne \boldsymbol{0}$:
$$h(\mathbf{p}, B) \triangleq \min_{\mathbf{x} \in \mathbb{R}_{\geq 0}^m :\, \langle \mathbf{p}, \mathbf{x} \rangle = B} d(\mathbf{x}).$$
For $\mp = \boldsymbol{0}$, we define $h(\boldsymbol{0}, B) = \sup_{\mx \in \-R^m_{\ge 0}} d(\mx)$. When the disutility $d$ is unbounded, we have $h(\boldsymbol{0}, B) = +\infty$. When $\mp \ne 0$, we also define the optimal demand set subject to earning constraint as $\demand(\mp)= \argmin_{\mx\in \-R^m_{\ge 0}: \InAngles{\mp, \mx} = B} d(\mx)$.

The indirect disutility function is analogous to the indirect utility function in the goods setting. We present several useful properties of the indirect disutility functions in the following lemmas. We first show that the indirect disutility function $h(\cdot, B)$ is non-increasing and quasi-concave.
\begin{lemma}
 \label{lem::chores-indirect-ni-qc}
 Suppose $d(\cdot)$ is non-decreasing. Then, $h(\cdot, B)$ is non-increasing and quasi-concave.
\end{lemma}

\begin{proof} We first show that $h( \cdot, B)$ is non-increasing. Note that $h(\boldsymbol{0},B) \ge h(\mp, B)$ for any $\mp \ne \boldsymbol{0}$. Now consider any $\mp \ge \mq \ge \boldsymbol{0}$ and $\mq \ne 0$. We let $\mx \in \demand(\mq, B)$ so that $h(\mq, B) = d(\mx)$ and $B = \InAngles{\mq, \mx} \le \InAngles{\mp, \mx}$. Let $\my = \frac{B}{\InAngles{\mp, \mx}} \mx \le \mx$ and thus $\InAngles{\mp, \my} = B$. Thus we have $h(\mp, B) \le d(\my) \le d(\mx) = h(\mq, B)$. This proves $h( \cdot, B)$ is non-increasing.

We now show $h(\cdot, B)$ is quasi-concave. Consider any price $\mp \ne \boldsymbol{0}$ and $\alpha \in (0,1)$, we have $h(\alpha\mp,B)\ge h(\mp, B) \ge \min\{h(\mp, B), h(\boldsymbol{0}, B) \}$ since $h(\cdot, B)$ is non-increasing. For any two price vector $\mp^{(1)}, \mp^{(2)} \ne \boldsymbol{0}$, and their demand $\mx^{(1)} \in \demand(\mp^{(1)}, B)$ and $\mx^{(2)} \in \demand(\mp^{(2)}, B)$. We have $h(\mp^{(1)}, B) = d(\mx^{(1)})$ and $h(\mp^{(2)}, B) = d(\mx^{(2)})$.
Fix any $\alpha \in (0,1)$, we let $\mp^{(3)} = \alpha \mp^{(1)} + (1 - \alpha) \mp^{(2)}$ and $\mx^{(3)} \in \demand(\mp^{(3)}, B)$. 
It suffices to prove that $h(\mp^{(3)}, B) = d(\mx^{(3)}) \ge \min\{d(\mx^{(1)}), d(\my^{(2)})\}$. Suppose not and we have $d(\mx^{(3)}) < \min\{d(\mx^{(1)}), d(\mx^{(2)})\}$. Since $\InAngles{\mp^{(3)}, \mx^{(3)}} = \alpha \langle \mathbf{p}^{(1)}, \mathbf{x}^{(3)} \rangle + (1 - \alpha) \langle \mathbf{p}^{(2)}, \mathbf{x}^{(3)} \rangle = B$, we have at least one of $\langle \mathbf{p}^{(1)}, \mathbf{x}^{(3)} \rangle \ge B$ or $\langle \mathbf{p}^{(2)}, \mx^{(3)} \rangle \ge B$ holds. Assume, w.l.o.g., that $\langle \mathbf{p}^{(1)}, \mathbf{x}^{(3)} \rangle \ge B$. Let $\mathbf{y} = \frac{B}{\langle \mathbf{p}^{(1)}, \mathbf{x}^{(3)} \rangle} \mathbf{x}^{(3)} \le \mx^{(3)}$. Then $\langle \mathbf{p}^{(1)}, \mathbf{y} \rangle = B$ and thus $d(\mx^{(3)}) \ge  d(\my) \ge h(\mp^{(1)}, B) = d(\mx^{(1)})$, which contradicts our previous assumption that $d(\mx^{(3)}) < d(\mx^{(1)})$. The claim follows by contradiction. 
\end{proof}

Similar to the goods setting, we establish a duality between direct and indirect disutility functions.
\begin{lemma}[Duality between Direct and Indirect Disutility Functions]\label{lem:chores:indirect-duality}
If $d$ is convex and non-decreasing, then, for any $\mx \in \-R^m_{\ge 0} \setminus \{\boldsymbol{0}\}$ and $B > 0$,
    $$d(\mx) = \max_{\mp \in \-R_{\geq 0}^m: \InAngles{\mp, \mx} = B} h(\mp, B).$$
\end{lemma}
\begin{proof}
  First, for any $\mp \in \-R^m_{\ge 0}$ such that $\InAngles{\mp, \mx} = B$, we have $h(\mp, B) = \min_{\my\in \-R_{\geq 0}^m: \InAngles{\mp,\my} = B} d(\my) \le d(\mx)$. This gives $d(\mx) \ge \max_{\mp \in \-R_{\geq 0}^m: \InAngles{\mp, \mx} = B} h(\mp, B)$. 
    
    Now we prove the other direction. We first note that if $d(\mx) = d(\boldsymbol{0})$, then the claimed inequality holds since for any $\mp \in \-R^m_{\ge 0}$, $h(\mp, B) \ge d(\boldsymbol{0}) = d(\mx)$. Now we consider $\mx$ such that  $d(\mx) > d(\boldsymbol{0}) \ge 0$. Let $\mg \in \partial d(\mx)$ be a subgradient. Since $d$ is non-decreasing, we know $\mg \ge \boldsymbol{0}$. By convexity, $d(\boldsymbol{0}) \ge d(\mx) + \InAngles{\mg, \boldsymbol{0} - \mx}$ which implies $\InAngles{\mg, \mx} \ge d(\mx) - d(\boldsymbol{0}) > 0$ and thus $\mg \ge \boldsymbol{0}$. Now we define a price vector $\mq:= B \frac{\mg}{ \InAngles{\mg, \mx}}$, which guarantees $\InAngles{\mq, \mx} = B$. Moreover, for any $\my \in \-R^m_{\ge 0}$ such that $\InAngles{\mq, \my} = B$, we have $d(\my) \ge d(\mx) + \InAngles{\mg, \my - \mx} = d(\mx)$ by convexity of $d$ and $\InAngles{\mg,\my - \mx} = \frac{\InAngles{\mg, \mx}}{B} \InAngles{\mq, \my - \mx} =0$. Thus we have  $h(\mq, B) = \min_{\my \in \-R_{\geq 0}^m: \InAngles{\mq, \my} = B} d(\my) = d(\mx)$. This implies $\max_{\mp \in \-R_{\geq 0}^m: \InAngles{\mp, \mx} = B} h(\mp, B) \ge h(\mq,  B) = d(\mx)$. 

    Combining the above two inequalities proves the claimed equality.
\end{proof}

\paragraph{Properties of 1-Homogeneous Disutilities}
We now restrict our attention to disutility functions that are convex, non-decreasing, and 1-homogeneous. We observe that for the problem of computing CE with 1-homogeneous and convex disutility functions, \textbf{it is without loss of generality to assume that every non-zero allocation of chores leads to non-zero disutility to every agent.}
\begin{assumption}\label{assumption:chores}
    For any $i \in N$, the disutility functions satisfy $d_i(\mx) > 0$ for any $\mx \ne \boldsymbol{0}$.
\end{assumption}
We include a detailed justification of why the assumption is without loss of generality in \Cref{sec:assumption-chores}. The intuition is that if there exists $i \in N$ and $\my \ne 0$ with $\+Y = \{j \in M: y_j > 0\}$ such that $d_i(\my) = 0$, we can allocate all the chores $j \in \+Y$ to agent $i$, remove the chores in $\+Y$ from the problem and work on the smaller instance with chores $M \setminus \+Y$. This procedure does not affect the disutility of agent $i$. We can continue the procedure until every non-zero allocation of chores leads to non-zero disutility to every agent. In the following, we assume without loss of generality that \Cref{assumption:chores} holds.

\paragraph{Dual Disutility} Similar to the goods setting, we introduce the notion of a dual disutility function. For disutility $d$ and $B > 0$, we define its dual disutility $\dual{d}(\cdot) = \frac{1}{h(\cdot, B)}$ using the indirect disutility function. We remark that since $d(\cdot)$ is non-decreasing and 1-homogeneous (behaving as a norm), its indirect disutility has an interesting connection to the \textbf{dual norm} $d^*(\mathbf{p})$, which is defined as:
\begin{align*}d^*(\mathbf{p}) \triangleq \max_{\mathbf{x} \in \mathbb{R}_{\geq 0}^m :\, d(\mathbf{x}) \leq 1} \langle \mathbf{p}, \mathbf{x} \rangle.\numberthis \label{eq::dual-norm}\end{align*}
The relationship is given by $d^*(\mathbf{p}) = \frac{B}{h(\mathbf{p}, B)}$ for any $B > 0$. This follows because $d$ is $1$-homogeneous and both optimization problems are equivalent to $\max_{\mx \ge \boldsymbol{0}} \frac{\langle \mathbf{p}, \mathbf{x} \rangle}{d(\mathbf{x})}$.

\paragraph{A New Program for Chores}
We now present our program for computing CE for chores with general non-decreasing, convex, and $1$-homogeneous disutilities. The high-level idea of obtaining our program is to first construct a dual public chores market and its Nash social welfare maximization program, and then convert it back to a program for the Fisher market by exchanging the role of allocations and prices. 

Given the fisher market instance $\+F = \{N, M, \{d_i\}_i, \{B_i\}_i\}$, we consider the dual instance of Lindahl equilibrium $\+L = \{N, M, \{\dual{d}_i\}_i, \{B_i\}_i\}$. Recall that the dual disutility is $\dual{d}_i(\cdot) = \frac{1}{h(\cdot, B)}$. Then the Nash social welfare program for the dual market is 
\begin{align*}
\min_{\mx \in \mathbb{R}_{\geq 0}^m: \langle \mx, \mathbf{1} \rangle = \sum_i B_i}\,&\, \sum_i B_i \log \beta_i \\
\text{s.t.} \,&\, \beta_i \geq 1 / h_i(\mx, B_i) , \text{ for all } i \in N 
\end{align*}
In the program, we minimize the product of agents' disutilities subject to the allocation constraints that $\sum_jx_j\ge \sum_i B_i$. Then, by exchanging the role of allocations and prices, we get a program for the original Fisher market. To handle potential $p_j = 0$ at the boundary, we define an extension of $h$ from $\mathbb{R}_{\geq 0}^m \times \mathbb{R}_{>0}$ to $\mathbb{R}^m \times \mathbb{R}_{>0}$ by letting:$$\hat{h}(\mathbf{p}, B) = h(\max\{\mathbf{0}, \mathbf{p}\}, B)$$where $\max\{\mathbf{0}, \mathbf{p}\}$ is the component-wise maximum. We note that $\hat{h}(\cdot, B)$ and $h(\cdot,B)$ is the same function over the whole domain of $\-R^m_{\ge 0}$. We introduce the extension solely for technical reasons.

The final optimization program for finding a Fisher market equilibrium for chores with non-decreasing, convex, and $1$-homogeneous disutilities:
\begin{align*}\min_{\mathbf{p} \in \mathbb{R}_{\geq 0}^m: \langle \mathbf{p}, \mathbf{1} \rangle = \sum_i B_i}\,&\, \sum_i B_i \log \beta_i \\
\text{s.t.} \,&\, \beta_i \geq 1 / \hat{h}_i(\mathbf{p}, B_i), \text{ for all } i \in N \tag{Fisher-Chores}\label{Fisher-Chores}
\end{align*}
The \ref{Fisher-Chores} program has no poles and opens up the possibility of applying fast first-order methods in the general setting. Moreover, the KKT points of \ref{Fisher-Chores} correspond directly to the Fisher market equilibrium.

\begin{theorem} \label{theorem: fisher-chores-KKT}
    For a Fisher market for chores $\{N, M, \{d_i\}_i, \{B_i\}_i\}$ with non-decreasing, convex, and $1$-homogenous disutilities, any Karush-Kuhn-Tucker (KKT) point of \ref{Fisher-Chores}  corresponds to a Fisher market equilibrium with chores.
\end{theorem}

\paragraph{Roy's Identity for Chores} As we have discussed in the technical overview, a key step in proving \Cref{theorem: fisher-chores-KKT} is to find the allocation $\{\mx_i\}$ given a price vector $\mp$ from a KKT point of \ref{Fisher-Chores}. In the linear case, \citep{chaudhury2024competitive} observed that the allocations $\{\mx_i\}$ exactly correspond to the dual variables of the \ref{EG-Chores-Linear-Dual}. However, this observation no longer holds for general convex and $1$-homogeneous disutilities as we only have $N$ constraints on the disutilities. In the following, we use the duality between indirect and direct disutility (\Cref{lem:chores:indirect-duality}) to establish a Roy's identity-like characterization of the demand (\Cref{lemma:Roy identity Chores}).

We first present some properties of $h$ and its extension $\hat{h}$.
\begin{lemma}\label{lem::convex::1h}
Let $d(\cdot)$ be 1-homogeneous, convex, non-decreasing.
\begin{enumerate}
    \item The function $h(\mp, B)$ is (-1)-homogeneous in $\mp \geq 0$.
    \item $1 / \hat{h}(\mp, B)$ is convex and non-decreasing in $\mp \in \-R^{m}$.
\end{enumerate} 
\end{lemma}
\begin{proof}
     \textbf{(1)}  We first show $h(\cdot, B)$ is (-1)-homogeneous. For any $\mp \ne 0$, we have $h(\mp, B) = \min_{\mx \in \-R^m_{\ge 0}, \InAngles{\mp, \mx} = B} d(\mx) > 0$. Consider any $c > 0$, we have $h(c \mp, B) = \min_{\mx \in \-R^m_{\ge 0}, c\InAngles{\mp, \mx} = B} d(\mx) = \min_{\my \in \-R^m_{\ge 0}, \InAngles{\mp, \my} = B} d(\frac{\my}{c}) = \frac{1}{c} \cdot  \min_{\my \in \-R^m_{\ge 0}, \InAngles{\mp, \my} = B} d(\my) = \frac{1}{c} h(\mp, B)$. 

    \textbf{(2)} Since $h(\cdot, B)$ is quasi-concave (Lemma~\ref{lem::chores-indirect-ni-qc}) and (-1)-homogeneous, $1 / h(\cdot, B)$ is a quasi-convex and $1$-homogeneous, which further implies it is convex over $\-R^m_{\ge 0}$. Now we show $\hat{h}(\cdot, B)$ is convex on $\-R^m$: for any $\mp, \mq \in \-R^m$ and $\alpha \in (0,1)$, we have
    \begin{align*}
        \alpha \cdot \frac{1}{\hat{h}(\mp, B)} + (1 - \alpha)  \cdot \frac{1}{\hat{h}(\mq, B)} 
        & = \alpha \cdot \frac{1}{h(\max\{\mp, \mathbf{0}\}, B)} + (1 - \alpha)  \cdot \frac{1}{h(\max\{\mq, \mathbf{0}\}, B)} \\
        & \geq  \frac{1}{h(\alpha  \max\{\mp, \mathbf{0}\} + (1 - \alpha) \max\{\mq, \mathbf{0}\}, B)}  \\
        & \geq  \frac{1}{h(\max\{\alpha \mp+ (1 - \alpha) \mq, \mathbf{0}\}, B)} \\
        &=  \frac{1}{\hat{h}(\alpha \mp + (1 - \alpha) \mq, B)},
    \end{align*}
    where the first inequality holds by convexity of $\frac{1}{h(\cdot, B)}$, the second inequality holds because $\max\{\ma, \boldsymbol{0}\} + \max\{\mb, \boldsymbol{0}\} \ge \max\{\ma + \mb, \boldsymbol{0}\}$ and that $\frac{1}{h(\cdot, B)}$ is non-decreasing.
\end{proof}
The main technical lemma of this section is the following Roy's identity-like characterization of the demand given a price $\mp$. Specifically, we show that for any subgradient $\mg \in \partial_{\mp} (\frac{1}{h(\mp, B)})$, $\mx:= B\frac{\mg}{\InAngles{\mg, \mp}} = Bh(\mp, B)\mg$ is an optimal demand under $\mp$, that is, $\mx \in \demand(\mp, B) = \argmin_{\my \in \-R^m_{\ge 0}, \InAngles{\mp, \my} = B} d(\my)$. Here $\InAngles{\mg, \mp} = \frac{1}{h(\mp, B)}$ since $\frac{1}{h(\cdot,B)}$ is $1$-homogeneous and Euler's homogeneous function theorem.
\begin{lemma}[Roy's Identity for Chores]\label{lemma:Roy identity Chores}
    Let $d(\cdot)$ be 1-homogeneous, convex, non-decreasing. 
    \begin{enumerate}
    \item For any price $\mp \in \-R^m_{\ge 0} \setminus \{\boldsymbol{0}\}$ and a subgradient $\mg$ of $1 /  \hat{h}(\mp, B)$ with regard to $\mp$, then $\mx = B \frac{\mg}{\InAngles{\mg, \mp}} = B h(\mp, B) \mg$ is an optimal demand under price $\mp$. In other words, $\mx \in \demand(\mp, B)$.
    \item For any price $\mp \in \-R^m_{\ge 0} \setminus \{\boldsymbol{0}\}$ an $\mx \in \demand(\mp, B)$, we have $\mg = \mx / (B \cdot \hat{h}(\mp, B))$ is a subgradient of $1 / \hat{h}(\mp, B)$ with regard to $\mp$.
\end{enumerate} 
\end{lemma}
\begin{proof}
    \textbf{(1)} Fix any $\mp \in \-R^m_{\ge 0} \setminus \{\boldsymbol{0}\}$. Let $\mg$ be a subgradient of $1 / \hat{h}(\mp, B)$. Since $\frac{1}{\hat{h}(\cdot, B)}$ is convex by \Cref{lem::convex::1h}, we have for any $\mq \ge \boldsymbol{0}$,  
    \begin{align}
      \frac{1}{\hat{h}(\mq, B)} \ge  \frac{1}{\hat{h}(\mp, B)} + \InAngles{\mg, \mq - \mp}  \Rightarrow \frac{1}{h(\mq, B)} \ge \frac{1}{h(\mp, B)} + \InAngles{\mg, \mq - \mp}  \label{eq:convex:h}
    \end{align}
    where the first inequality implies the second by the definition of $\hat{h}$ and that both $\mp, \mq \ge \boldsymbol{0}$. 
    
    We remark that $\mg$ satisfies the following two properties: \begin{enumerate} 
    \item $\mg \geq 0$. This is because $\frac{1}{\hat{h}(\cdot, B)}$ is non-decreasing by \Cref{lem::convex::1h};
    \item $\InAngles{\mg, \mp} > 0$. This is because if $\InAngles{\mg, \mp} = 0$, then by \eqref{eq:convex:h}, $1 / h(\mq, B) \geq 1 / h(\mp, B)$ for any $\mq = \alpha \mp$ for any $\alpha$,  and this leads to a contradiction as $1 / h(\mp, B)$ is a $1$-homogeneous (\Cref{lem::convex::1h}) and non-zero function over $\-R^m_{\ge 0} \setminus \{\boldsymbol{0}\}$ (\Cref{assumption:chores}).
    \end{enumerate}
    
    Now we show that the allocation $\mx := B \frac{\mg}{\InAngles{\mg, \mp}} = B h(\mp, B) \mg$ is an optimal demand under price $\mp$.  We have $\InAngles{\mp, \mx} = B$ by definition. Moreover, for any other $\mq \in \-R^m_{\ge 0}$ such that $\InAngles{\mq, \mx} = B$, we have $\InAngles{\mg, \mq - \mp} = \frac{\InAngles{\mg, \mp}}{B} \InAngles{\mx, \mq - \mp}  = 0$, which implies 
    $h(\mq, B) \leq h(\mp, B)$ from \eqref{eq:convex:h}. By Lemma~\ref{lem:chores:indirect-duality}, $d(\mx) = h(\mp, B)$ and $\InAngles{\mp, \mx} = B$, which implies $\mx \in \demand(\mp, B)$ is an optimal demand.

    \textbf{(2)} Fix any $\mp \in \-R^m_{\ge 0} \setminus \{\boldsymbol{0}\}$ and an optimal demand $\mx = \demand(\mp, B)$, we claim that $\mg := \frac{\mx}{B \cdot \hat{h}(\mp, B)} = \frac{\mx}{B \cdot h(\mp, B)} \ge \boldsymbol{0}$ is a subgradient of $\frac{1}{\hat{h}(\mp, B)}$. That is, we need to show that 
    \begin{align}\label{eq:allocation->subgradient}
        \frac{1}{\hat{h}(\mq, B)} \ge \frac{1}{\hat{h}(\mp, B)} + \InAngles{\mg, \mq - \mp} , \forall \mq \in \-R^m.
    \end{align}

    We first show that \eqref{eq:allocation->subgradient} holds for any $\mq \geq 0$. By definitions of $\hat{h}(\cdot, B)$ and $\mg = \frac{\mx}{B \cdot h(\mp, B)}$, we know that \eqref{eq:allocation->subgradient} is equivalent to 
    \begin{align*}
        \frac{1}{h(\mq, B)} \ge \frac{\InAngles{\mx, \mq}}{B \cdot h(\mp, B)} 
    \end{align*}
    We consider two cases. In the case where $\InAngles{\mx, \mq} = 0$,  the above inequality thus \eqref{eq:allocation->subgradient} holds immediately. In the remaining case where $\InAngles{\mx, \mq} > 0$, we define $\hat{\mq} = \frac{B}{\InAngles{\mq, \mx}} \mq$. Since $\InAngles{\hat{\mq}, \mx} = B$, we have $h(\mp, B) \geq h(\hat{\mq}, B)$ by Lemma~\ref{lem:chores:indirect-duality}. Since $h$ is (-1)-homogeneous (\Cref{lem::convex::1h}), this further implies $h(\mp, B) \geq \frac{\InAngles{\mq, \mx}}{B} \cdot  h(\mq, B)$ and thus the equivalent condition \eqref{eq:allocation->subgradient}. This completes the proof of \eqref{eq:allocation->subgradient} for $\mq \ge 0$. 
    
    Now we show that \eqref{eq:allocation->subgradient} holds for $\mq \ne \-R^m_{\ge 0}$ by reducing it to the previous case. Fix any $\mq \ne \-R^m_{\ge 0}$,  we have
    \begin{align*}
        \frac{1}{\hat{h}(\mq, B)} = \frac{1}{h(\max\{\mq, \boldsymbol{0}\}, B)} \ge \frac{1}{h(\mp, B)} + \InAngles{\mg, \max\{\mq, \boldsymbol{0}\}- \mp} \ge \frac{1}{\hat{h}(\mp, B)} + \InAngles{\mg, \mq - \mp},
    \end{align*}
    where the first inequality holds since $\max\{\mq, \boldsymbol{0}\} \ge \boldsymbol{0}$ and \eqref{eq:allocation->subgradient} holds for such a vector $\max\{\mq, \boldsymbol{0}\}$; the second inequality holds by definition of $\hat{h}$ and the facts that $\mg \ge 0$ and $\max\{\mq, \boldsymbol{0}\} \ge \mq$.
\end{proof}

\begin{proof}[Proof of \Cref{theorem: fisher-chores-KKT}]
    Let $\{\mp, \{\beta_i\}_{i \in [N]}\}$ be a KKT point of \ref{Fisher-Chores}. Let $\{\lambda_i\}_{i \in N}$ be the multipliers for the inequality constraints $\{\frac{1}{\hat{h_i}(\mp, B)} - \beta_i \le 0 \}$, $\{\gamma_j\}_{j \in M}$ be the multipliers for the inequality constraints $\{-p_j \le 0\}$, $\{\alpha_i\}_{i \in N}$ be the multipliers for the inequality constraints $\{- \beta_i \le 0\}$, and $\mu$ be the multiplier for the equality constraint $\sum_{j} p_j = \sum_i B_i$. By KKT conditions, there exist subgradients $\{\mg_i \in \partial_\mp (\frac{1}{\hat{h}_i(\mp,B)})\}$ such that
    \begin{itemize}
        \item[1.] $p_j \ge 0$ for $j \in M$, $\beta_i > 0$ for  $i \in N$;
        \item[2.] $\sum_j p_j = \sum_i B_i$ and $\beta_i \ge \frac{1}{\hat{h}_i(\mp, B)}$ for all $i \in N$;
        \item[3.] $\lambda_j \ge 0$ and  $\gamma_j \ge 0$ for $j \in M$; $\alpha_i \ge 0$ for $i \in N$;
        \item[4.] $\gamma_j p_j =0$ for  $j \in M$; $\lambda_i (\frac{1}{\hat{h}_i(\mp, B)} - \beta_i)=0$ and $\alpha_i \beta_i = 0$ for $i \in N$;
        \item[5.] $\sum_i \lambda_i \mg_{ij} - \gamma_j + \mu  = 0$ for $j \in M$;
        \item[6.] $\frac{B_i}{\beta_i} - \lambda_i - \alpha_i = 0$ for $i \in N$.
    \end{itemize}
    We remark that (1) and (2) contain primal feasibility conditions; (3) contains dual feasibility conditions; (4) contains the complementary slackness conditions; (5) and (6) contain the stationarity conditions. By conditions (1), (3), (4), we have $\alpha_i = 0$ for all $i \in N$.  By conditions (1) and (6), we have $\lambda_i = \frac{B_i}{\beta_i} > 0$, which further implies $\beta_i = \frac{1}{\hat{h}_i(\mp, B)}$ by conditon (4) and thus $\lambda_i = B_i\cdot\hat{h}_i(\mp,B)$. 

    Now we define $\mx_i = \lambda_i \mg_i =  B_i \hat{h}(\mp, B) \mg_i$. We claim that $\{\mp, \{\mx_i\}\}$ is a CE. We note that by \Cref{lemma:Roy identity Chores}, $\mx_i \in \argmin_{\my_i \in \-R^m_{\ge 0}: \InAngles{\mp, \my_i} = B_i} d_i(\my_i)$ satisfies the earning constraint and the disutility-minimizing condition for all $i \in N$. It remains to show the market-clearing condition that $\sum_j x_{ij} = 1$ for all $j \in M$. We note that if $p_j > 0$, then by (4) we have $\gamma_j = 0$, and then by (5), we have 
    $\sum_ix_{ij} = -\mu$; if $p_j = 0$, then by (3) we have $\lambda_j \ge 0$ and then by (4) we have $\sum_j x_{ij} = -\mu + \lambda_j \ge -\mu$. With this in mind,  we have
    \begin{align*}
        \sum_i B_i = \sum_i \sum_j p_j x_{ij} = \sum_j p_j \cdot \InParentheses{\sum_i x_{ij}} = \InParentheses{\sum_{j: p_j >0} p_j} \cdot (-\mu) = - \mu \cdot \sum_j p_j
    \end{align*}
    Since by (2) we have $\sum_j p_j = \sum_i B_i$, we have $-\mu = 1$. Thus for any $j$ such that $p_j > 0$, we have $\sum_j x_{ij} = -\mu = 1$. The only remaining issues are $\{j: p_j =0\}$. For these chores, we have $\sum_i x_{ij}=-\mu+\lambda_j \ge 1$. But since $p_j = 0$ and $\{d_i\}$ are non-decreasing, we can always reduce the allocation of these chores to $1$ without affecting the earning condition or the disutility-minimizing constraint for any agent\footnote{This is also intuitive since even before the modification, $\mx_i$ already minimizes disutility subject to the earning constraint for each agent $i$}. After the modification, $\{\mp,\{\mx_i\}\}$ becomes a CE.

    Conversely, given a CE $(\mp, \{\mx_i\})$, we can set the subgradient $\mg_i = \frac{\mx_i}{B \cdot \hat{h}_i(\mp, B)}$ (\Cref{lemma:Roy identity Chores}), $\beta_i = 1 / \hat{h}_i(\mp, B_i)$,  $\lambda_i = \frac{B_i}{\beta_i}$, $\gamma_j = 0$, $\mu = -1$, and $\alpha_i = 0$. All KKT conditions are satisfied.
\end{proof}

\subsection{Lindahl Equilibrium with Public Chores}\label{sec:public chores}

In contrast to private chores, literature on the fair division of \emph{public chores} (where a single allocation level is applied to all agents) is very sparse. A primary reason may be the lack of meaningful fairness notions; for example, since every agent "consumes" the same bundle of public chores, the concept of envy-freeness is vacuous.

Drawing on the duality between private and public goods, we propose the notion of a Lindahl equilibrium for public chores. In this framework, $\mathcal{L} = (N, M, \{d_i\}, \{B_i\})$, each agent $i \in N$ has an earning constraint $B_i > 0$ and a non-decreasing disutility function $d_i: \mathbb{R}_{\geq 0}^m \to \mathbb{R}_{\geq 0}$. The equilibrium consists of a single public chore allocation $\mathbf{x}$ and a set of personalized prices $\{\mathbf{p}_i\}$ for each agent.

\begin{definition} [Lindahl Equilibrium for Chores]
Given a Lindahl equilibrium instance $\mathcal{L}$, the allocation and personalized prices $(\mathbf{x}, \{\mathbf{p}_i\})$ form a Lindahl equilibrium if the following hold:
\begin{enumerate}
    \item \emph{Disutility Minimization}: The allocation $\mathbf{x}$ is affordable and disutility-minimizing for every agent at their personalized prices:
    $\mathbf{x} \in \argmin_{\mathbf{x}' \in \mathbb{R}_{\geq 0}^m: \langle \mathbf{p}_i, \mathbf{x}' \rangle = B_i} d_i(\mathbf{x}')$ for every agent $i \in N$.
    \item \emph{Price Feasibility}: The personalized prices sum to a "market" price (normalized to 1 for each chore):
    For every chore $j \in M$, $\sum_{i \in N} p_{ij} = 1$.
\end{enumerate}
\end{definition}
We show that a Lindahl equilibrium allocation is always weakly Pareto-optimal for non-decreasing disutilities and is Pareto-optimal for increasing disutilities.
\begin{definition}[Pareto-Optimality] An allocation $\mx$ is \emph{Pareto-optimal} if there is no other feasible allocation $\my$ such that $\sum_j y_j = \sum_i B_i$, and $d_i(\my) \le d_i(\mx)$ for every agent $i \in N$ with at least one strict inequality.  An allocation $\mx$ is \emph{weakly Pareto-optimal} if there is no other feasible allocation $\my$ such that  $d_i(\my) < d_i(\mx)$ for every agent $i \in N$.
\end{definition}

\begin{theorem}
    Consider a public chores instance $\+L = \{N, M, \{d_i\}, \{B_i\}\}$ with non-decreasing disutilities. If $(\mx, \{\mp_i\})$ is a Lindahl equilibrium, then $\mx$ is weakly Pareto-optimal. If additionally each $d_i$ is increasing, then $\mx$ is Pareto-optimal.
\end{theorem}
\begin{proof}
    Suppose $\mx$ is not weakly Pareto-optimal and there exists $\my \in \-R^m$ such that $\sum_j y_j = \sum_i B_i$, and $d_i(\my) < d_i(\mx)$ for every agent $i \in N$. Then for every agent $i \in N$, we must have $\InAngles{\mp_i, \my} < B_i$ by the disutility-minimizing constraint of the Lindahl equilibrium. This implies $\sum_{i} \InAngles{\mp_i, \my} < \sum_i B_i$. However, this leads to the following contradiction:
    \begin{align*}
        \sum_i B_i > \sum_{i} \InAngles{\mp_i, \my} = \sum_j \sum_i p_{ij} y_j = \sum_{j} y_j = \sum_i B_i,
    \end{align*}
    where we use the price feasibility condition of the Lindahl equilibrium: $\sum_i p_{ij} = 1$ for each $j \in M$. 

    Now let us consider the case where each $d_i$ is increasing. Suppose that $\mx$ is no Pareto-optimal and there exists $\my$ such that $\sum_j y_j = \sum_i B_i$, and $d_i(\my) \le d_i(\mx)$ for every agent $i \in N$ with at least one strict inequality. We claim that $\InAngles{\mp_i, \my} \le B_i$ for each $i \in N$ with at least one strict inequality. This is because if $\InAngles{\mp_i, \my} > B_i$, the allocation $\my' = \frac{\InAngles{\mp_i, \my}}{B_i} \my$ satisfies $\InAngles{\mp_i, \my'} = B_i$ and $d_i(\my') < d_i(\my) \le d_i(\mx)$ violating the disutility-minimization property of $\mx$. Thus we get $\sum_{i} \InAngles{\mp_i, \my} < \sum_i B_i$, which leads to the contraction that $\sum_i B_i > \sum_i B_i$ by steps in the first case.
\end{proof}

Similar to the case of goods, we introduce the concept of dual disutility.
\begin{definition}
    Given a disutility function $d$ and budget $B > 0$, its dual disutility is $\dual{d}$ such that $\dual{d}(\mx) = \frac{1}{h(\mx, B)}$, where $h$ is the indirect disutility associated with $d$.
\end{definition}
We note that when $d(\cdot)$ is homogeneous and convex, the dual disutility is closely related to the dual norm \eqref{eq::dual-norm}.

Building on Lemma~\ref{lem:chores:indirect-duality}, we establish a direct equivalence between the Lindal equilibrium for public chores and a Fisher market equilibrium with chores. 
\begin{theorem}\label{theorem:chores-duality}
    Consider a public chores market $\+L = \{N, M, \{d_i\}, \{B_i\}\}$ where $\{d_i\}$ is convex and non-decreasing. Define the dual private chores market $\mathcal{F} = (N, M, \{\dual{d}_i\}, \{B_i\})$ where $\dual{d}_i$ are the dual utilities of $d_i$. Then $(\mx, \{\mp_i\})$ is a Lindahl equilibrium of $\+L$ if and only if $(\{\dual{\mx}_i\}, \dual{\mp})$ is a Fisher market equilibrium of $\+F$, where $\dual{\mx}_i = \mp_i$ and $\dual{\mp} = \mx$.
\end{theorem}

\paragraph{A Program for Lindahl Equilibrium of Public Chores}
We now examine the specific case where the disutility functions $d_i$ are homogeneous, convex, and non-decreasing.

To find Lindahl equilibrium of $\+L = \{N, M, \{d_i\}_i, \{B_i\}_i\}$, we consider the dual Fisher market for private chores $\+F = \{N, M, \{\dual{d}_i\}_i, \{B_i\}_i\}$. Then we can apply \ref{Fisher-Chores} on the dual Fisher market whose KKT points correspond to Fisher market equilibria (\Cref{theorem: fisher-chores-KKT}). In the dual market, the indirect disutility of the dual disutility $\dual{d}$, $\dual{h}_i$, simplifies to the reciprocal of the original disutility function $d_i$:
\begin{align*}
    \dual{h}_i(\dual{\mp}, B_i) = \min_{\dual{\mx} \in \-R_{\geq 0}^m: \InAngles{\dual{\mp}, \dual{\mx}}} \dual{d}(\dual{\mx}) = \min_{\mp \in \-R_{\geq 0}^m: \InAngles{\dual{\mp}, \mp}} \frac{1}{h_i(\mp, B_i)} = \frac{1}{\max_{\mp \in \-R_{\geq 0}^m: \InAngles{\dual{\mp}, \mp}} h_i(\mp, B_i)} = \frac{1}{d_i(\dual{\mp})}.
\end{align*}
This allows us to characterize the Lindahl equilibrium using the KKT points of the following program: note that we just exchange the role of allocations and prices in \ref{Fisher-Chores}
\begin{align*}\min_{\mathbf{x} \in \mathbb{R}_{\geq 0}^m: \langle \mathbf{x}, \mathbf{1} \rangle = \sum_i B_i}\,&\, \sum_i B_i \log \beta_i \\
\text{s.t.} \,&\, \beta_i \geq \hat{d}_i(\mx). \tag{Lindahl-Chores}\label{Lindahl-Chores}
\end{align*}
where $\hat{d}_i$ is the extension of $d_i$ by letting: $\hat{d}(\mathbf{x}) = d(\max\{\mathbf{0}, \mathbf{x}\}, B)$ where $\max\{\mathbf{0}, \mathbf{x}\}$ is the component-wise maximum. By \Cref{theorem: fisher-chores-KKT} and \Cref{theorem:chores-duality}, the KKT points of \ref{Lindahl-Chores} correspond to the Lindahl equilibria of $\+L$.
\section{Missing Proofs and More Discussions}\label{sec::missing}

\subsection{Proofs of \Cref{theorem:duality-utility-indirec-utility} and \Cref{lemma:equivalence between fisher (ii)}}\label{sec::direct-indirect}
\begin{proof}[Proof of \Cref{theorem:duality-utility-indirec-utility}] 
    Fix any $\mx \in \-R^m_{+} \setminus \{\boldsymbol{0}\}$ and $B > 0$. For any $\mp \in \-R^m_{+}$ with $\InAngles{\mx, \mp} \le B$, we have $v(\mp, B) = \max_{\my \in \-R^m_{+}, \InAngles{\my, \mp} \le B} u(\my) \ge u(\mx)$. This proves one direction $\min_{\mp \in \-R^m_{+}, \InAngles{\mp, \mx}\le B} v(\mp, B) \ge u(\mx)$.

    Now we prove the other direction. Consider any small enough $\varepsilon >0$ such that there exists $\my \in \-R^m_{+}$ with $u(\my) = u(\mx) + \varepsilon$ (the existence of $\varepsilon > 0$ is guaranteed since $u$ is locally non-satiated). Define the superlevel set  $U_\alpha := \{\my\in \R^m_{+}: u(\my) \ge \alpha\}$ where $\alpha = u(\mx) + \varepsilon$. Since $u$ is locally nonsatiated, continuous, and quasi-concave, we know $U_\alpha$ is non-empty, closed, and convex.  Since $\mx \notin U_\alpha$, we can find a separating hyperplane $\mp_\varepsilon \in \-R^m$ that separates $\mx$ from $U_\alpha$. That is, $\InAngles{\mp_\varepsilon, \my} > t \ge \InAngles{\mp_\varepsilon, \mx}$ for any $\my \in U_\alpha$. Clearly $\mp_\varepsilon \ne \boldsymbol{0}$. Moreover, for any $j \in [m]$, if $\my \in U_\alpha$, then $\my + \beta \me_j \in U_\alpha$ for any $\beta \ge 0$ since $u$ is non-decreasing. This implies $\InAngles{ \mp_\varepsilon, \my + \beta \me_j} = \InAngles{\mp_\varepsilon, \my} + \beta \mp_\varepsilon[j] > t$ holds for all $\beta \ge 0$ and thus we have $\mp_\varepsilon[j] \ge 0$ for all $j \in [m]$. Now given $\mp_\varepsilon \ge \boldsymbol{0}$ and $\mp_\varepsilon \ne \boldsymbol{0}$, we conclude that we can find $t > 0$ such that $\InAngles{\mp_\varepsilon, \my} > t \ge \InAngles{\mp_\varepsilon, \mx}$ for any $\my \in U_\alpha$. We can then properly scale $\mp_\varepsilon$ so that $\InAngles{\mp_\varepsilon, \my} > B \ge \InAngles{\mp_\varepsilon, \mx}$ for all $\my \in U_\alpha$. This implies $v(\mp_\varepsilon, B)  = \sup_{\my \in \-R^m_{+}, \InAngles{\my, \mp_\varepsilon}\le B} u(\my) \le \alpha = u(\mx) + \varepsilon$ and thus  $\min_{\mp \in \-R^m_{+}, \InAngles{\mp,\mx} \le B} v(\mp, B) \le v(\mp_\varepsilon, B) \le u(\mx) + \varepsilon$. Taking $\varepsilon \rightarrow 0$, we get $\min_{\mp \in \-R^m_{+}, \InAngles{\mp, \mx} \le B} v(\mp, B) \le u(\mx)$.

    Combining the above two inequalities gives the equality $u(\mx) = \min_{\mp \in \-R^m_{+}, \InAngles{\mx, \mp}\le B} v(\mp, B)$.
\end{proof}

\begin{proof}[Proof of \Cref{lemma:equivalence between fisher (ii)}]
    "$\Rightarrow$": If \eqref{cond::Lindahl-equilibrium} holds, then $u_i(\mx) = \max_{\my\in \-R^m_{\ge 0}: \InAngles{\my, \mp_i} \le B_i} u_i(\my) = v_i(\mp_i, B_i)$. For any $\mq_i \in \-R^m_{\ge 0}$ such that $\InAngles{\mx, \mq_i} \le B_i$, we have $v_i(\mq_i, B_i) = \max_{\my \in \-R^m_{\ge 0}: \InAngles{\my, \mq_i} \le B_i} u_i(\my) \ge u_i(\mx) = v_i(\mp_i, B_i)$. This implies \eqref{cond::Lindahl-dual}.

    "$\Leftarrow$": If \eqref{cond::Lindahl-dual} holds, using \Cref{theorem:duality-utility-indirec-utility}, we have $v_i(\mp_i, B_i) = \min_{\mq_i \in \-R^m_{\ge 0}: \InAngles{\mx, \mq_i}\le B_i} v_i(\mq_i, B_i) = u_i(\mx)$. By definition of the indirect utility $v_i(\mp_i, B_i)$, it implies $\mx \in \argmax_{\my \in\-R^m_{\ge 0}: \InAngles{\my, \mp_i}\le B_i} u_i(\my)$ thus \eqref{cond::Lindahl-equilibrium} holds.
\end{proof}

\subsection{Proof of \Cref{lem::duality-homogeneous}}\label{sec:proof duality-homogeneous}
\begin{proof}[Proof of \Cref{lem::duality-homogeneous}]
    Fix any $B > 0$. We note that $\dual{u}(\boldsymbol{0}) = \frac{1}{v(\boldsymbol{0}, B)} = \frac{1}{\infty} = 0$. 

    We prove that $v(\cdot, B)$ is decreasing. Let $\mp > \mq \ge \boldsymbol{0}$. Note that $v(\mp, B) \le u(\frac{B}{\min_j p_j}\boldsymbol{1}) = \frac{B}{\min_j p_j} u(\boldsymbol{1}) < +\infty$. If $v(\mq, B) = +\infty$, we are done. We consider the case where $v(\mq, B) < +\infty$. Let $v(\mp, B) = u(\mx)$. Then we have $\InAngles{\mp, \mx} \le B$ and the demand $\mx' = \frac{\InAngles{\mp, \mx}}{\InAngles{\mq,\mx}}\mx$ satisfies $\InAngles{\mq, \mx'} \le B$ and we have $v(\mq, B) \ge u(\mx') = \frac{\InAngles{\mp, \mx}}{\InAngles{\mq,\mx}} \cdot u(\mx) > u(\mx)$. This proves that $v(\cdot, B)$ is decreasing and thus $\dual{u}(\cdot)$ is increasing.
    
    For any $\mp \ne 0$. If $v(\mp, B) = +\infty$, then $v(t\cdot \mp, B) = +\infty$ for any $t > 0$, which implies $\dual{u}(\mp) = t \cdot \dual{u}(\cdot \mp) = 0$ for all $t > 0$.  If $v(\mp, B) =u(\mx) \ne +\infty$, since $u$ is $1$-homogeneous, we have $v(t\cdot \mp, B) = u(\frac{1}{t} \cdot \mx) = \frac{1}{t} u(\mx) = \frac{1}{t} v(\mp, B)$ for any $t > 0$, which imlies $\dual{u}(t\cdot\mp) = t \cdot \dual{u}(\mp)$ for any $t > 0$. Thus $\dual{u}$ is $1$-homogeneous.

    Since $\dual{u}$ is $1$-homogeneous and increasing, it is clear that $\dual{u}$ is continuous. 
    
    Since $\dual{u}$ is quasi-concave, increasing, and $1$-homogeneous, $\dual{u}$ is also concave by \Cref{thm: quasi-concave-2-concave-new}. 
\end{proof}

\subsection{Proof of \Cref{theorem::non-homogenous} and \Cref{example::etoe}} \label{sec:non-homogenous-proof}
\begin{proof}[Proof of \Cref{theorem::non-homogenous}]
Let $\-B \triangleq \sum_{i} B_i$.
    Note that since $(\mx, \{\mp_i\})$ is a Lindahl equilibrium, then $u_i(\mx) \geq \min \left\{1, \frac{B_i}{\InAngles{\mp_i,\my}} \right\}u_i(\my)$, as agent $i$ can use their budget, $B_i$, to buy at least $\frac{B_i}{\InAngles{\mp_i,\my}}$ fraction of $\my$. The lower bound holds because if $B_i\leq \InAngles{\mp_i,\my}$, then $u_i\left(\frac{B_i}{\InAngles{\mp_i,\my}} \my\right)\ge \frac{B_i}{\InAngles{\mp_i,\my}}u_i(\my)$  using concavity and $u_i(\mathbf{0})=0$. 
Thus,
    \begin{align*}
        \sum_{i} \frac{B_i}{\-B} \log \frac{u_i(\my)}{u_i(\mx)} \leq  \sum_{i} \frac{B_i}{\-B} \log \max \left\{1, \frac{\InAngles{\mp_i,\my}}{B_i} \right\}.
    \end{align*}
    Let $B'_i \triangleq {\InAngles{\mp_i,\my}}$, and let $\mathcal{A}'$ be the set of agents such that $B'_i \geq B_i$, $B \triangleq \sum_{i \in \mathcal{A}'} B_i$, and $B' \triangleq \sum_{i \in \mathcal{A}'} B'_i$. Therefore, using the log sum inequality,
    \begin{align*}
           \sum_{i} \frac{B_i}{\-B} \log \frac{u_i(\my)}{u_i(\mx)} \leq \sum_{i\in \+A'} \frac{B_i}{\-B} \log \frac{B'_i}{B_i} \leq \frac{B}{\-B} \log \frac{B'}{B} \, .
    \end{align*}
   By definition, $B \leq \-B$. We also note that $B'\le \sum_{i} B_i'=\sum_{i} \InAngles{\mp_i,\my} \le \sum_{j} y_j  \leq \sum_i B_i = \-B$. 
    \begin{align*}
         \sum_{i} \frac{B_i}{\-B} \log \frac{u_i(\my)}{u_i(\mx)} \leq \max_{\frac{B}{\-B} \leq 1} \frac{B}{\-B} \log \frac{\-B}{B} \leq \frac{1}{e}\, .
    \end{align*}
    The statement follows.
\end{proof}

\begin{example}\label{example::etoe}To demonstrate the tightness of the bound in \Cref{theorem::non-homogenous}, consider a Lindahl economy with two public goods, $\mathbf{x} = (x_1, x_2)$, and two agents. Agent $1$ has a linear utility function $u_1(\mathbf{x}) = x_1$ and a budget $B_1 = 1$. Agent $2$ has a concave utility function and a budget $B_2 = e-1$. The utility function is given by:$$u_2(\mathbf{x}) = \min\left\{\frac{x_1}{e} + \frac{x_2}{e-1}, \epsilon x_1 + \frac{\epsilon e}{e-1} x_2 + (1 - \epsilon e), 1 + \frac{\epsilon}{e-1} x_2, \epsilon x_1 + \frac{\epsilon}{e-1} x_2 + (1 - \epsilon)\right\}.$$

We analyze two different Lindahl equilibria: (i) Equilibrium 1: The personalized prices are $\mathbf{p}_1 = (1, 0)$ and $\mathbf{p}_2 = (0, 1)$, and the allocation is $\mathbf{x} = (1, e-1)$. At this allocation, the utilities are $u_1(1, e-1) = 1$ and $u_2(1, e-1) = 1 + \epsilon$. The resulting NSW is $\left(1^1 \cdot (1+\epsilon)^{e-1}\right)^{1/e} = \left((1+\epsilon)^{e-1}\right)^{1/e}$. (ii) Equilibrium 2: Prices: The personalized prices are $\mathbf{p}_1 = (1/e, 0)$ and $\mathbf{p}_2 = (1 - 1/e, 1)$, and the allocation is $\mathbf{x} = (e, 0)$. At this allocation, the utilities are $u_1(e, 0) = e$ and $u_2(e, 0) = 1$. The resulting NSW is $\left(e^1 \cdot 1^{e-1}\right)^{1/e} = e^{1/e}$. Comparing the NSW values from these two equilibria demonstrates that the bound in \Cref{theorem::non-homogenous} is tight.\end{example}

\subsection{Proof of \Cref{lem::GS-NG}} \label{sec::GS-NG-lemma-proof}
\begin{proof}[Proof of \Cref{lem::GS-NG}]
    Let the scaling factor $\alpha \triangleq \max_k\left\{\frac{p_k}{p'_k}\right\}$ and $\alpha \geq 1$. We start with price vector $\mp$ and budget $B$. Scaling prices and budget by $\alpha$ leaves demand unchanged:
    $$\demandnb_j(\mp, B) = \demandnb_j(\mp / \alpha, B / \alpha).$$
    Then, we compare $(\mp / \alpha, B / \alpha)$ to $(\mp', B)$. We note $(\mp / \alpha, B / \alpha) \leq (\mp', B)$, and $\mp_j / \alpha = \mp'_j$. Therefore, we can increase the price vector and budget from $(\mp / \alpha, B / \alpha)$ to $(\mp', B)$. By gross substitutes and normal goods properties, this only increases the demand for good $j$.
\end{proof}

\subsection{Proof of \Cref{theorem:PRD:lindahl:tc}} \label{sec::proof:theorem:PRD:lindahl:tc}
We restate PRD for the Lindahl equilibrium where agents have total complements utilities:

The dynamic updates agent $i$'s spending contribution to good $j$ as:$$b_{ij}^{t+1} = p_{ij}^t \cdot \demandnb_{ij}(\mathbf{p}_i^t, B_i)$$where $\demandnb_{ij}(\mathbf{p}_i^t, B_i)$ is agent $i$'s Marshallian demand given their personalized prices $\mathbf{p}_i^t$ and total budget $B_i$. This update rule has a clear economic interpretation as an expenditure best-response. At each iteration $t$, agent $i$ observes their personalized prices $\mathbf{p}_i^t$ and sets their next-period budget $b_{ij}^{t+1}$ to their current optimal expenditure. 

The personalized Lindahl price $p_{ij}^t$ is determined by the agents' contributions. It represents agent $i$'s proportional share of the total spending on good $j$, $x_j^t = \sum_k b_{kj}^t$: $p_{ij}^t = b_{ij}^t / x_j^t = b_{ij}^t / (\sum_k b_{kj}^t)$.

\begin{proof}[Proof of \Cref{theorem:PRD:lindahl:tc}]
Let $(\mx^*, \{\mp^*_i\})$ denote a Lindahl equilibrium. Let $b_{ij}^* \triangleq x_j^* p_{ij}^*$ be the equilibrium spending of agent $i$ on good $j$. By the total complements condition, $\mx^* > 0$. Additionally, $\mx^t > 0$ by the updating rule, and the utility function is strictly increasing. 

The convergence analysis focuses on the KL divergence between the equilibrium spending and the spending at time $t+1$:
\begin{align*}
\sum_{ij} b_{ij}^* \log \frac{b_{ij}^*} {b_{ij}^{t+1}}  &= \sum_{ij} b_{ij}^* \log \frac{b_{ij}^*}{p_{ij}^t \demandnb_{ij}(\mp_i^t, B_i)} \\
&= \sum_{ij} b_{ij}^* \log \frac{b_{ij}^*}{b_{ij}^t} - \sum_{ij} b_{ij}^* \log \frac{x_j^*}{x_j^t} + \sum_{ij} b_{ij}^* \log \frac{\demandnb_{ij}(\mp_i^*, B_i)}{\demandnb_{ij}(\mp_i^t, B_i)}.
\end{align*}

To prove the convergence, we show that this potential function is non-increasing. The core of the argument lies in establishing the following inequality for any $\{\mp_i\}$ such that $\sum_i p_{ij} = 1$.
\begin{align*}
    \sum_{j} b_{ij}^* \log \frac{\demandnb_{ij}(\mp_i^*, B_i)}{\demandnb_{ij}(\mp_i, B_i)} \leq \sum_j x^*_{j} \cdot (p_{ij} - p^*_{ij}) \numberthis \label{eq::tc::tech::1}
\end{align*}

For convenience, let $\mx \triangleq \demand_{i}(\mp_i, B_i)$ and $\mx^* = \demand_{i}(\mp_i^*, B_i)$. Given any nonnegative allocation $\mx$, let $\mq_i(\mx)$ denote the price vector such that $\mx$ is the demand. Since the utility function is strictly concave, strictly increasing, and satisfies the total complements property, this price vector always exists and is unique.

We prove \eqref{eq::tc::tech::1} by considering an adjustment procedure that transforms an initial allocation $\mx$ into the final allocation $\mx^*$ through a series of intermediate steps. In each step, the current allocation $\mx'$ is updated to a new allocation vector $\mx''$ by applying one of the following two operations:
\begin{enumerate}
    \item Operation $1$: Let $\+S$ be the set of the goods with the maximal ratio $\frac{x^*_j}{x'_j} > 1$, if any. \emph{Increase} the $x'_j$ to $x''_j$ proportionally for all $j \in \+S$ until the ratio $\frac{x^*_j}{x''_j}$ matches the second maximal ratio or $1$.
     \item Operation $2$: Let $\+S$ be the set of the goods with the minimal ratio $\frac{x^*_j}{x'_j} < 1$, if any. \emph{Decrease} the $x'_j$ to $x''_j$ proportionally for all $j \in \+S$ until the ratio $\frac{x^*_j}{x''_j}$ matches the second minimal ratio or $1$.
\end{enumerate}
Our goal is to show that for both operations, the following inequality holds \begin{align*}
    \sum_{j} b_{ij}^* \log \frac{x''_j}{x'_j} \leq \sum_j x^*_{j}\cdot [q_{ij}(\mx') - q_{ij}(\mx'')] \numberthis \label{eq::tcd::tech::2}
\end{align*}
If this holds for every step, then \eqref{eq::tc::tech::1} follows by summing the terms over all steps and performing a telescoping sum.

Let's focus on Operation $1$ (Operation $2$ is analogous).

Recall $\+S$ denotes the set of goods with the largest ratio between $\mx^*$ and $\mx'$, let $r'$ ($r' > 1$) denote this ratio, $r' \triangleq \max_j \frac{x^*_j}{x'_j}$ and let $r''$ denote the largest ratio for goods not in $\+S$ (or $1$ if all the goods are in $\+S$). In this step:
\begin{enumerate} \item For $j \in \mathcal{S}$, $x''_j = x'_j \cdot (r'/r'')$, so $\log(x''_j / x'_j) = \log(r'/r'')$. \item For $j \notin \mathcal{S}$, $x''_j = x'_j$, so $\log(x''_j / x'_j) = 0$. \item Crucially, for $j \in \mathcal{S}$, $x^*_j = r' x'_j = r'' x''_j$. For $j \notin \mathcal{S}$, $x_j^* \le r'' x'_j = r'' x''_j$.
\end{enumerate}

First, by total complements condition, $\sum_{j \in \+S} b_{ij}^* = \sum_{j \in \+S} x_j^* p_{ij}^* \leq \sum_{j \in \+S} (r'' x''_j)\cdot q_{ij}(\mx'')$. Therefore, \eqref{eq::tcd::tech::2} is implied by the following inequality:
$$r''\sum_{j \in \+S}  x''_j q_{ij}(\mx'') \log (r' / r'') \leq \sum_j x^*_{j}\cdot [q_{ij}(\mx') - q_{ij}(\mx'')].$$
Additionally, the RHS can be lower bounded by $\sum_j r'' x''_j \cdot [q_{ij}(\mx') - q_{ij}(\mx'')]$ as $x^*_{j} = r'' x''_j$ for $j \in \+S$ and, for $j \notin \+S$, $x^*_j \leq r'' x''_j$ and $q_{ij}(\mx') \leq q_{ij}(\mx'')$ by total complements condition. Additionally, $\sum_j r'' x''_j \cdot [q_{ij}(\mx') - q_{ij}(\mx'')] = \sum_j r'' [x_j' q_{ij}(\mx') - x_j'' q_{ij}(x'')] + \sum_{j \in \+S} (r' - r'') x'_j q_{ij}(\mx')= (r' - r'') \sum_{j \in \+S} x'_j q_{ij}(\mx')$.
Thus, to prove \eqref{eq::tcd::tech::2}, we only need to show 
$$r''\sum_{j \in \+S}  x''_j q_{ij}(\mx'') \log (r' / r'')  \leq (r' - r'') \sum_{j \in \+S} x'_j q_{ij}(\mx').$$
The inequality follows as $r'' \log (r' / r'') \leq r' - r''$ and  $\sum_{j \in \+S}  x''_j q_{ij}(\mx'') \leq \sum_{j \in \+S} x'_j q_{ij}(\mx')$ as $\sum_{j \notin \+S}  x''_j q_{ij}(\mx'') \geq \sum_{j \notin \+S} x'_j q_{ij}(\mx')$ by the total complements condition.

This confirms that the KL divergence is non-increasing and $\sum_t \sum_{ij} b_{ij}^* \log \frac{x_j^*}{x_j^t} = \sum_t \sum_{j} x_j^* \log \frac{x_j^*}{x_j^t}$ is upper bounded, which implies $\mx^t \rightarrow \mx^*$ and the empirical convergence of $x_j^t$ is $O(1 / T)$.
\end{proof}

\subsection{Proof of \Cref{thm:prd-lindahl-gs}} \label{sec::proof:prd-lindahl-gs}
We restate PRD for the Lindahl equilibrium where agents have utilities satisfying gross substitutes and normal goods properties: 

The dynamic for agent $i$'s spending contribution to good $j$ is:\begin{align*} b_{ij}^{t+1} = B_i \frac{x^t_{j} \nabla_j u_i(\mx^t)}{\sum_{j'} x^t_{j'} \nabla_{j'} u_i(\mx^t) }\quad \text{where } p^t_{ij} = \frac{b_{ij}^t}{x_j^t} \text{ and } x_j^t = \sum_i b_{ij}^t.\end{align*}

The following proof relies on several lemmas, which are stated in Appendix~\ref{sec:properties-GS}. The proofs for these lemmas can be found in \cite{cheung2025proportional}.
\begin{proof}[Proof of \Cref{thm:prd-lindahl-gs}]
We first define an agent-specific price vector $\mq_i(\mx)$ for agent $i$ with allocation $\mx > 0$ such that $q_{ij}(\mx) = B_i \frac{\nabla_j u_i(\mx)}{\sum_{j'} x_{j'} \nabla_j u_i(\mx) }$. By construction, this price vector $\mq_i(\mx)$ is such that the allocation $\mx$ is optimal for agent $i$ under this price vector: $\mx = \argmax_{\my \in \-R_{\geq 0}^m: \InAngles{\mq_i(\mx), \my} \leq B_i} u_i(\my)$.

Let $(\mx^*, \{\mp^*_i\})$ be the Lindahl equilibrium. According to the definition of the Lindahl equilibrium (\Cref{def::Lindahl-equilibrium}), $\sum_j x_j^* = \sum_j x_j^* \sum_i p_{ij}^* = \sum_i B_i$ and, according to our updating rule, $\sum_j x_j^t = \sum_{ij} b_{ij}^t = \sum_i B_i$. Additionally, as the utility functions $u_i$ are strictly increasing, all equilibrium prices must be strictly poisitive, $\mp^*_i > 0$ for any agent $i$, to prevent agents from demanding an infinite amount of goods.

We now analyze the change in the Kullback-Leibler (KL) divergence between the equilibrium $\mx^*$ and the allocation $\mx^{t+1}$:
\begin{align*}\sum_j x_j^* \log \frac{x_j^*}{x_j^{t+1}} &= \sum_j x_j^* \log \frac{x_j^*}{\sum_i x_j^{t} q_{ij}(\mx^t)} \\ &=\sum_j x_j^* \log \frac{x_j^*}{x_j^t} + \sum_j x_j^* \log \frac{\sum_i p_{ij}^*}{\sum_i q_{ij}(\mx^t)}.\end{align*}

We now show $\sum_j x_j^* \log \frac{\sum_i p_{ij}^*}{\sum_i q_{ij}(\mx^t)} \leq 0$. 

According to the definition of the Lindahl equilibrium (\Cref{def::Lindahl-equilibrium}), $x_j^* = x_j^* (\sum_i p_{ij}^*)$, and the log-sum inequality 
\begin{align*}
\sum_j x_j^* \log \frac{\sum_i p_{ij}^*}{\sum_i q_{ij}(\mx^t)} = \sum_j x_j^* \left(\sum_i p_{ij}^*\right)  \log \frac{\sum_i p_{ij}^*}{\sum_i q_{ij}(\mx^t)} \leq \sum_{ij} x_j^*  p_{ij}^*  \log \frac{p_{ij}^*}{q_{ij}(\mx^t)}.
\end{align*}

By \Cref{pec:lem:main-technique}, we have the following inequality for each agent $i$:  
\begin{align*}
    \sum_{j} x_j^*  p_{ij}^*  \log \frac{p_{ij}^*}{q_{ij}(\mx^t)}\leq \sum_j p_{ij}^* \left( x_j^t - x_j^*\right). \numberthis \label{eq::strict-negative-potential}
\end{align*}

Summing these inequalities over all agents $i$ gives 
\begin{align*}\sum_{ij} x_j^*  p_{ij}^*  \log \frac{p_{ij}^*}{q_{ij}(\mx^t)} \leq \sum_{ij} p_{ij}^* \left( x_j^t - x_j^*\right) = \sum_j \left( x_j^t - x_j^*\right) \sum_i p_{ij}^*  \leq 0.\end{align*}

This confirms that the KL divergence between $\mx^*$ and $\mx^{t}$ is non-negative and monotonically decreasing. 

To show $\mx^t$ converges to $\mx^*$, we argue by contradiction. Assume  $\bbx^t$ does not converge to $\bbx^*$.   Since $\bbx^t \in [0, \sum_i B_i]^m$, there must exists a subsequence $\bbx^{t_k}$ that converges to some $\tilde{\bbx} \neq \bbx^*$. We analyze two cases for the limit point $\tilde{\bbx}$.

\paragraph{Case 1: The limit point $\tilde{\bbx}$ is stictly positive ($\tilde{\bbx} > 0$)} If $\tilde{\mx} > 0$, the price function $\mq_i(\mx)$ is continuous at $\tilde{\mx}$. Since $\tilde{\mx} \neq \mx^*$, \Cref{pec:lem:tech-1} and \Cref{pec:lem:main-technique}, there exists a $\delta > 0$ such that 
\begin{align*}
    \sum_{j} p^*_{ij} x_{j}^* \log \frac{p^*_{ij}}{q_{ij}(\tilde{\bbx})}  - \sum_{j} p^*_{ij}(\tilde{x}_{j} - x_j^*) \leq -\delta. \numberthis \label{pec:eq:main-2-1}
\end{align*}
 Thus, for a sufficiently large $t_k$,  this implies: 
\begin{align*}\sum_{ij} x_{j}^* p_{ij}^*  \log \frac{p_{ij}^*}{q_{ij}(\mx^{t_k})}  &\leq \sum_{ij} p^*_{ij} x_j^* \log \frac{p^*_{ij}}{q_{ij}(\bbx^{t_k})}  - \sum_{ij} p^*_{ij}(x^{t_k}_{j} - x_j^*) \\
&\rightarrow \sum_{ij} p^*_{ij} x_j^* \log \frac{p^*_{ij}}{q_{ij}(\tilde{\bbx})}  - \sum_{ij} p^*_{ij}(\tilde{x}_{j} - x_j^*) \\
&\overset{(a)}{\leq} -\delta.
\end{align*} 
Here, (a) is from \eqref{pec:eq:main-2-1} and the limit holds due to \Cref{pec:lem:converge-allocation-1}.  This means the KL divergence decreases by at least $\delta / 2$ infinitely often. This leads to a contradiction as the KL divergence is non-negative.

\paragraph{Case 2: the limit point $\tilde{\bbx}$ has zero components} In this case, we introduce a projected price vector $\tilde{\mq}_i^t$ for each agent, projecting  $\bbq_i(\bbx^{t_k})$ onto the bounded domain $[\epsilon \bbp_i^*, \frac{1}{\epsilon} \bbp_i^*]$ for some small $\epsilon > 0$:
\begin{align*}
\tilde{q}_{ij}^t = \begin{cases}
    \epsilon p^*_{ij} &\text{ if } q_{ij}(\bbx_i^t) <  \epsilon p^*_{ij}\\
    \frac{1}{\epsilon}  p^*_{ij} &\text{ if } q_{ij}(\bbx_i^t) >  \frac{1}{\epsilon} p^*_{ij}\\
    q_{ij}(\bbx_i^t) &\text{ o.w.}
\end{cases},
\end{align*} where $\epsilon$ will be specified later. 
    
    By \Cref{pec:cor:main-technique}, 
    \begin{align*}&\sum_j p^*_{ij} x_j^* \log \frac{p^*_{ij}}{q_{ij}(\bbx^t)}  - \sum_j p^*_{ij}\left[x^t_j - x^*_j\right] \leq \sum_j p^*_{ij} x_j^* \log \frac{p^*_{ij}}{\tilde{q}^t_{ij}}  - \sum_j p^*_{ij}\left[\demandnb_{ij}(\tilde{\bbq}^t_i, B_i) - x_j^*\right].  \numberthis \label{eq::gs-converg-1}
    \end{align*}
    Since  $[\epsilon \bbp^*_i, \frac{1}{\epsilon} \bbp^*_i]$ is bounded, there exists a further subsequence of subsequence $\{\bbx^{t_k}\} \rightarrow \tilde{\bbx}$ such that $\tilde{\bbq}^{t_k}_i \rightarrow \tilde{\bbq}^+_i$. Note that
    \begin{align*}\sum_{ij} x_{j}^* p_{ij}^*  \log \frac{p_{ij}^*}{q_{ij}(\mx^{t_k})}  &\leq \sum_{ij} p^*_{ij} x_j^* \log \frac{p^*_{ij}}{q_{ij}(\bbx^{t_k})}  - \sum_{ij} p^*_{ij}(x^{t_k}_{j} - x_j^*) \\
&\leq \sum_{j} p^*_{ij} x_j^* \log \frac{p^*_{ij}}{q_{ij}(\bbx^{t_k})}  - \sum_{j} p^*_{ij}(x^{t_k}_{j} - x_j^*) \\
    & \leq \sum_j p^*_{ij} x_j^* \log \frac{p^*_{ij}}{\tilde{q}^{t_k}_{ij}}  - \sum_j p^*_{ij}(\demandnb_{ij}(\tilde{\bbq}^{t_k}_i, B_i) - x_j^*)\quad \text{(by \eqref{eq::gs-converg-1})} \\
    &\rightarrow \sum_j p^*_{ij} x_j^* \log \frac{p^*_{ij}}{\tilde{q}^{+}_{ij}}  - \sum_j p^*_{ij}(\demandnb_{ij}(\tilde{\bbq}^+_i, B_i) - x_j^*).
    \end{align*}
    The second inequality holds as, for any agent $i$, $\sum_{j} p^*_{ij} x_j^* \log \frac{p^*_{ij}}{q_{ij}(\bbx^{t_k})}  - \sum_{j} p^*_{ij}(x^{t_k}_{j} - x_j^*) \leq 0$ according to \Cref{pec:lem:main-technique}. Therefore, we can only focus on one particular agent $i$. 
    
    Note that if one can show $\demand_i(\tilde{\bbq}^+_i, B_i) \neq \bbx^*$, the RHS is no larger than $-\delta$ for some positive $\delta$ (by \Cref{pec:lem:main-technique}), which makes this case invalid similar to Case $1$.

    Assume for a contradiction that $\demand_i(\tilde{\bbq}^+_i, B_i) = \bbx^*$ but $\tilde{\bbx} \neq \bbx^*$ for small enough $\epsilon$. We choose $\epsilon$ small enough such that $\tilde{\bbq}^+_i > 2 \epsilon \bbp_i^*$ (guaranteed by \Cref{pec:lem:tech-2}). 
    In this case, for sufficiently large $t_k$, 
    \begin{enumerate}
        \item For all good $j$ with $x_{j}^* > 0$, as $\demand_i(\tilde{\bbq}^+_i, B_i) = \bbx^*$, by \Cref{pec:lem:tech-1}, $\tilde{q}^{t_k}_{ij} \rightarrow \tilde{q}^+_{ij} = p^*_{ij}$. For sufficiently large $t_k$, this implies $\tilde{q}^{t_k}_{ij} \in (\epsilon \bbp_i^*, \frac{1}{\epsilon} \bbp_i^*)$. Therefore, $q_{ij}(\bbx^{t_k}) = \tilde{q}^{t_k}_{ij} \rightarrow \tilde{q}^+_{ij} = p^*_{ij}$;
        \item For all good $j$, $q_{ij}(\bbx^{t_k}) \geq \tilde{q}^{t_k}_{ij}  \rightarrow \tilde{q}^+_{ij}$  (as $\tilde{\bbq}^+_i > 2 \epsilon \bbp_i^*$ by \Cref{pec:lem:tech-2}).
    \end{enumerate}
    
    Thus, for sufficiently large $t_k$, for any $j$ with $x_{j}^* > 0$, 
    $$ x_{j}^{t_k} = \demandnb_{ij}(\bbq_i(\bbx^{t_k}), B_i) \geq \demandnb_{ij}(\tilde{\bbq}^{t_k}_i, B_i) \rightarrow \demandnb_{ij}(\tilde{\bbq}^{+}_i, B_i) = x_{j}^*.$$
    The inequality holds by gross substitutes property as $q_{ij}(\bbx^{t_k}) = \tilde{q}^{t_k}_{ij}$ for $j$ with $x_{j}^* > 0$ and $q_{ij}(\bbx^{t_k}) \geq \tilde{q}^{t_k}_{ij}$ for all $j$. 
    
    Therefore,   \begin{align*}B_i &\geq \sum_{j: x_{j}^* > 0} q_{ij}(\bbx^{t_k}) \demandnb_{ij}(\bbq_i(\bbx^{t_k}), B_i) \geq  \sum_{j: x_{j}^* > 0} q_{ij}(\bbx^{t_k}) \demandnb_{ij}(\tilde{\bbq}^{t_k}_i, B_i) \\
    & \quad \quad \quad \quad= \sum_{j: x_{j}^* > 0} \tilde{q}^{t_k}_{ij} \demandnb_{ij}(\tilde{\bbq}^{t_k}_i, B_i) \overset{(a)}{\rightarrow} \sum_{j: x_{j}^* > 0} p_{ij}^* x_{j}^* = B_i. \numberthis \label{ineq::9}\end{align*}
    Here, (a) is from \Cref{pec:lem:converge-allocation-1} and $\demand_i(\tilde{\bbq}^+_i, B_i) = \bbx^*$. This implies $\sum_{j: x_{j}^* = 0} q_{ij}(\bbx^{t_k}) \demandnb_{ij}(\bbq_i(\bbx^{t_k}), B_i) \rightarrow 0$, which yields 
    \begin{enumerate} 
    \item $x_j^{t_k} = \demandnb_{ij}(\bbq_i(\bbx^{t_k}), B_i) \rightarrow 0$ when $x_{j}^* = 0$ as $q_{ij}(\bbx^{t_k}) \geq \tilde{q}_{ij}^{t_k} \rightarrow \tilde{q}_{ij}^+ > \epsilon p_{ij}^*$ for sufficiently large $t_k$ and $p_{ij}^* > 0$;
    \item and, in \eqref{ineq::9}, as $q_{ij}(\bbx^{t_k}) \rightarrow p_{ij}^*$ for $x_{j}^* > 0$, $x_{j}^{t_k} = \demandnb_{ij}(\bbq_i(\bbx^{t_k}), B_i) \rightarrow x_{j}^*$.
    \end{enumerate}
    This contradicts the assumption that $\tilde{\bbx} \neq \bbx^*$, which implies that the allocation $\mx^t$ converges to $\mx^*$.
\end{proof}
\subsection{Discussions on profit-maximizing condition in Lindahl equilibrium} \label{sec::general-lindahl-discussion}
We begin by recalling the definition of a Lindahl equilibrium.
\begin{definition}[Lindahl equilibrium]Given $\+L$, let $\mx$ be an allocation and $\{\mp_i\}_{i \in N}$ be nonnegative personalized prices. Then $(\mx, \{\mp_i\}_{i \in N})$ is a \emph{Lindahl equilibrium} if 
\begin{enumerate}[label=(\roman*)]
    \item $\mx$ is \emph{affordable}: $\InAngles{\mp_i, \mx} \leq B_i$ for every agent $i \in N$,
    \item $\mx$ is \emph{utility-maximizing}: $\mx \in \argmax_{\my \in \-R_{\geq 0}^m: \InAngles{\mp_i, \my} \leq B_i} u_i(\my)$ for every agent $i \in N$,
    \item $\mx$ is \emph{profit-maximizing}: for every $j \in M$, $\sum_i p_{ij} \leq 1$; and whenever $x_j >0$ then $\sum_i p_{ij} = 1$.
\end{enumerate}
\end{definition}
We draw particular attention to the profit maximization condition (iii). This condition, utilized in much of the literature  \cite{kroer2025computing, fain2016core, 10.1093/restud/rdaf043} is equivalent to the case there is a producer of public goods maximizes the profit subject to unit cost of each good. Specifically,  condition (iii) is equivalent to $$\mx \in \argmax_{\my \in \-R_{\geq 0}^m} \sum_j \left(\sum_i p_{ij}\right) y_j - \sum_j y_j.$$

This model can be generalized to accommodate a producer with a more complex, non-linear cost:

$$\mx \in \argmax_{\my \in \-R_{\geq 0}^m} \sum_j \left(\sum_i p_{ij}\right) y_j - c(\my),$$ where $c(\my)$ is a convex function represents the cost of producing $\my$.

\begin{definition}[Lindahl equilibrium with general cost function] \label{def::Lindahl-equilibrium-producer} Given $\+L$, let $\mx$ be an allocation and $\{\mp_i\}_{i \in N}$ be nonnegative personalized prices. Then $(\mx, \{\mp_i\}_{i \in N})$ is a \emph{Lindahl equilibrium} if 
\begin{enumerate}[label=(\roman*)]
    \item $\mx$ is \emph{affordable}: $\InAngles{\mp_i, \mx} \leq B_i$ for every agent $i \in N$,
    \item $\mx$ is \emph{utility-maximizing}: $\mx \in \argmax_{\my \in \-R_{\geq 0}^m: \InAngles{\mp_i, \my} \leq B_i} u_i(\my)$ for every agent $i \in N$,
    \item $\mx$ is \emph{profit-maximizing}: $\mx \in \argmax_{\my \in \-R_{\geq 0}^m} \sum_j \left(\sum_i p_{ij}\right) y_j - c(\my)$.
\end{enumerate}
\end{definition}

This generalization of the Lindahl equilibrium corresponds to a generalized dual economy. We find that the dual counterpart is no longer a standard Fisher market, but rather a \emph{Fisher market with production}. In this dual market, the seller is not endowed with a fixed supply of goods. Instead, the seller produces goods $\my$ according to a convex cost function $s(\cdot)$, maximizing profit given market prices $\mp$. The market-clearing condition then requires that the total demand equals the seller's profit-maximizing supply.

\begin{definition}[Production Fisher market equilibrium]\label{def::Fisher-market-production} Given $\+F$, let $\{\mx_i\}_{i \in N}$ be a personalized allocations and $\mp$ be a price vector. Then $(\{\mx_i\}_{i \in N}, \mp)$ is a \emph{Fisher market equilibrium} if
\begin{enumerate}[label=(\roman*)]
    \item $\mx_i$ is affordable: $\InAngles{\mp, \mx_i} \le B_i$ for every agent $i \in N$,
    \item $\mx_i$ is utility-maximizing: $\mx_i\in \argmax_{\my_i \in \-R_{\geq 0}^m: \InAngles{\mp, \my_i} \leq B_i} u_i(\my_i)$ for every agent $i \in N$,
    \item \emph{Seller profit-maximization and market clearing}: 
     $\my \in \argmax_{\my' \in \-R_{\geq 0}^m} \left( \sum_j y'_j p_j - s(\my') \right)$, where $\my = \sum_i \mx_i$.
\end{enumerate}
\end{definition}

The following theorem formalizes this new duality. It demonstrates that a Lindahl equilibrium with a general producer cost $c$ and agents' utilities $u_i$ is dual to a production Fisher market with a seller cost $s$ and agents; utilities $\dual{u}_i$. Crucially, the relationship between the two cost functions is defined by the Fenchel conjugate.

\begin{theorem}
    A Lindahl equilibrium with a general cost function, $(\mx, \{\mp_i\}_{i \in N})$, with agents' utility functions $\{u_i\}_{i \in N}$ and producer's cost function $c$, corresponds to a production Fisher market equilibrium, $(\{\dual{\mx}_i\}_{i \in N}, \dual{\mp})$, with agents' utility functions $\{\dual{u}_i\}_{i \in N}$ and seller's cost function $s$. The correspondence is given by $\mx = \dual{\mp}$, $\mp_i = \dual{\mx}_i$, $\dual{u}_i$ is the dual utility of $u_i$, and $s$ is the Fenchel conjugate of $c$:$$s(\my) = \sup_{\mp \in \-R_{\geq 0}^m} \{\mp^\top \my - c(\mp)\}.$$
\end{theorem}
\begin{proof}
Assume $(\mx, \{\mp_i\}_{i \in N})$ is a Lindahl equilibrium with cost $c$ (per Definition \ref{def::Lindahl-equilibrium-producer}). Let the dual variables be defined as $\dual{\mp} = \mx$, $\dual{\mx}_i = \mp_i$, $\dual{u}_i$ be the dual utility of $u_i$, and $s = c^*$ (the Fenchel conjugate of $c$).

As established by our main duality result, Theorem \ref{thm::main-duality}, the affordability (i) and utility-maximization (ii) conditions of the production Fisher market (Definition \ref{def::Fisher-market-production}) are satisfied by $(\{\dual{\mx}_i\}_{i \in N}, \dual{\mp})$. We only need to verify the seller profit-maximization and market-clearing condition (iii).

Let $\dual{\my}$ be the vector of total demand in the Fisher market, i.e., $\dual{y}_j = \sum_i \dual{x}_{ij}$ for all good $j \in M$. Using the duality mapping, $\dual{x}_{ij} = p_{ij}$, which implies $\dual{y}_j = \sum_i p_{ij}$.

 According to the \Cref{def::Lindahl-equilibrium-producer}, 
    $$\mx \in \argmax_{\my \in \-R_{\geq 0}^m} \sum_j \left(\sum_i p_{ij}\right) y_j - c(\my).$$

    This implies, for any $\dual{\mq} \in \-R_{\geq 0}^m$ 
    $$ \sum_j \left(\sum_i \dual{x}_{ij}\right) \dual{p}_j - c(\dual{\mp}) \geq \sum_j \left(\sum_i \dual{x}_{ij}\right) \dual{q}_j - c(\dual{\mq}).$$ 
    
    Recall that  $s(\my) = \max_{\mp \in \-R_{\geq 0}^m} \{\mp^\top \my - c(\mp)\}$. This implies
    $$s(\dual{\my}) = \sum_j \dual{y}_j \dual{p}_j - c(\dual{\mp}).$$

    Therefore, 
    $$\sum_j \dual{y}_j \dual{p}_j - s(\dual{\my}) = c(\dual{\mp}) \geq \sum_j \dual{y}'_j \dual{p}_j - s(\dual{\my}').$$
    The inequality holds as $s(\dual{\my}') = \max_{\mp \in \-R_{\geq 0}^m} \{\mp^\top \dual{\my}' - c(\mp)\} \geq \sum_j \dual{y}'_j \dual{p}_j - c(\dual{\mp})$. The result follows.
    
\end{proof}

\subsection{On \Cref{assumption:chores}}\label{sec:assumption-chores}
We show that it is without loss of generality to assume \Cref{assumption:chores} that every non-zero allocation of chores leads to non-zero disutility to every agent. Suppose there exists an agent $i$ and an allocation $\mathbf{y} \neq \mathbf{0}$ such that $d_i(\mathbf{y}) = 0$. Let $\mathcal{Y} = \{j \mid y_j > 0\}$ be the set of chores in the support of $\mathbf{y}$.

First, we show that $d_i(\mathbf{z}) = 0$ for any allocation $\mathbf{z}$ supported only on $\mathcal{Y}$ (i.e., $z_j = 0$ for all $j \notin \mathcal{Y}$). For any such $\mathbf{z}$, there exists a scalar $c \ge 0$ such that $\mathbf{z} \le c\mathbf{y}$ (component-wise). Since $d_i(\mathbf{0}) = 0$, by monotonicity and non-negativity, $0 \le d_i(\mathbf{z}) \le d_i(c\mathbf{y})$. By 1-homogeneity, $d_i(c\mathbf{y}) = c \, d_i(\mathbf{y}) = c \cdot 0 = 0$. Thus, $d_i(\mathbf{z}) = 0$.

Second, let $\mathbf{x}$ be any allocation supported \emph{outside} $\mathcal{Y}$ (i.e., $x_j = 0$ for $j \in \mathcal{Y}$), and let $\mathbf{z}$ be any allocation supported \emph{on} $\mathcal{Y}$, as above. Define $\mathbf{x}' = \mathbf{x} + \mathbf{z}$. We know $d_i(\mathbf{x}' - \mathbf{x}) = d_i(\mathbf{z}) = 0$. Using convexity, 1-homogeneity, and the fact that $d_i(\mathbf{x}' - \mathbf{x}) = 0$, we have:
$$ d_i(\mathbf{x}) = d_i(\mathbf{x}) + d_i(\mathbf{x}' - \mathbf{x}) \ge 2 d_i\left(\frac{\mathbf{x} + (\mathbf{x}' - \mathbf{x})}{2}\right) = 2 d_i\left(\frac{1}{2} \mathbf{x}'\right) = d_i(\mathbf{x}') $$
This gives $d_i(\mathbf{x}) \ge d_i(\mathbf{x}')$. Additionally, since $d_i$ is also non-decreasing, we have $d_i(\mathbf{x}') = d_i(\mathbf{x} + \mathbf{z}) \ge d_i(\mathbf{x})$. Combining these two inequalities, we must have $d_i(\mathbf{x}') = d_i(\mathbf{x})$.

This demonstrates that agent $i$ is indifferent to receiving any bundle of chores from the set $\mathcal{Y}$, as it does not change their (dis)utility. Therefore, we can assign all chores $j \in \mathcal{Y}$ to agent $i$ and remove these chores from the allocation problem. The WLOG assumption holds for the remaining problem instance.

\subsection{Proportional Response Dynamics for Lindahl Equilibrium with CES Utilities} \label{sec:prd:ces}
We analyze the proportional response dynamics for a Lindahl equilibrium by examining a generalized Shmyrev program, $\Phi(\mb)$, which is derived from the program \cite{shmyrev2009algorithm, cheung2018dynamics} for dual Fisher markets with dual CES utilities. This potential function is defined as:
\begin{align*}\Phi(\mb) = - \sum_{i: \rho_i \neq \{0, -\infty\}} &\frac{1}{\rho_i} \sum_j b_{ij} \log \frac{ b_{ij}}{a_{ij} x^{\rho_i}_j} - \sum_{i: \rho_i = -\infty} \sum_j b_{ij} \log \frac{a_{ij}}{ x_j} + \sum_{i : \rho_i = 0} \sum_j b_{ij} \log x_j, \\
&\text{ where } x_j = \sum_i b_{ij} \text{ and } \sum_j b_{ij} = B_i. 
\end{align*}
 The function satisfies the following inequality, which establishes its strong convexity and smoothness parameters relative to the KL divergence:
\begin{align*}
   &-\sum_{i: \rho_i \neq \{0, -\infty\}} \frac{1}{\rho_i} \texttt{KL}(\mb_i || \mb_i')  \\ 
   &~~~~~~~~~~~~~~~\leq \Phi(\mb) - \Phi(\mb') - \InAngles{\nabla \Phi(\mb'), \mb - \mb'} \\
   &~~~~~~~~~~~~~~~~~~~~~~~~~~~~~~\leq \sum_{i: \rho_i \neq \{0, -\infty\}} \frac{\rho_i - 1}{\rho_i} \texttt{KL}(\mb_i || \mb_i') + \sum_{i : \rho_i = -\infty} \texttt{KL}(\mb_i || \mb_i') .
\end{align*}
\paragraph{Complementary domain} When all agents have $\rho_i \leq 0$, then the function $\Phi(\cdot)$ is convex. We can therefore apply the standard mirror descent algorithm to find the equilibrium:
$$b_{ij}^{t+1} = \argmin_{\mb_i : \sum_j b_{ij} = B_i} \left\{ \InAngles{\nabla_{\mb_i} \Phi(\mb^t), \mb_i - \mb_i^t} + \frac{\rho_i - 1}{\rho_i} \texttt{KL}(\mb_i || \mb_i') \right\}$$
for $\rho_i \in (-\infty, 0)$, which yields
\begin{align*}b_{ij}^{t+1} = B_i \frac{\left(a_{ij} \cdot \left(x_{ij}^t / b_{ij}^t\right)^{\rho_i} \right)^{1 / (1 - \rho_i)}}{\sum_{j'}\left(a_{ij'} \cdot \left(x_{ij'}^t / b_{ij'}^t\right)^{\rho_i} \right)^{1 / (1 - \rho_i)}};  \numberthis\label{eq:prd:ces:lindahl:neg}  \end{align*}
 and, for $\rho_i  = -\infty$, apply the standard mirror descent:
$$b_{ij}^{t+1} = \argmin_{\mb_i : \sum_j b_{ij} = B_i} \left\{ \InAngles{\nabla_{\mb_i} \Phi(\mb^t), \mb_i - \mb_i^t} + \texttt{KL}(\mb_i || \mb_i') \right\}, $$
 which yields
\begin{align*}b_{ij}^{t+1} = B_i \frac{a_{ij} \cdot \left(b_{ij}^t / x_{ij}^t\right) }{\sum_{j'}a_{ij'} \cdot \left(b_{ij'}^t / x_{ij'}^t\right) }.\numberthis\label{eq:prd:ces:lindahl:leon}
\end{align*}
This leads to the following convergence theorem.
\begin{theorem}\label{theorem:Lindahl-CES-PRD-complement}
    Consider the Lindahl equilibrium with agents having CES utilities such that the elasticity parameter $\rho_i \leq 0$ for all $i$. The dynamics converge to the Lindahl equilibrium if the agents update their spendings according to \eqref{eq:prd:ces:lindahl:neg} if $\rho_i \in (-\infty, 0)$; and \eqref{eq:prd:ces:lindahl:leon} if $\rho_i = -\infty$; and $b_{ij}^{t} = B_i \frac{a_{ij}}{\sum_{j'} a_{ij'}}$ if $\rho_i = 0$; and $x_j^t  = \sum_i b_{ij}^t$ for all $j$. The dynamics converge to the Lindahl equilibrium at a rate of $O(1 / T)$. Additionally, if no agent has $\rho_i = -\infty$, the convergence is $O\left((1 - 1 / (1 - \max_i \{\rho_i\}))^T\right)$.
\end{theorem}

\paragraph{Substitute domain} Similarly, when all agents have $\rho_i \geq 0$, the function $\Phi(\cdot)$ is concave. We can apply the standard mirror ascent:
$$b_{ij}^{t+1} = \argmax_{\mb_i : \sum_j b_{ij} = B_i} \left\{ \InAngles{\nabla_{\mb_i} \Phi(\mb^t), \mb_i - \mb_i^t} - \frac{1}{\rho_i} \texttt{KL}(\mb_i || \mb_i') \right\}$$
for $\rho_i \in (0, 1]$, which yields
\begin{align*}b_{ij}^{t+1} = B_i \frac{a_{ij} \left(x_{ij}^t\right)^{\rho_i}}{\sum_{j'} a_{ij'} \left(x_{ij'}^t\right)^{\rho_i}}.  \numberthis\label{eq:prd:ces:lindahl:pos}  \end{align*}
 This implies the following theorem.
\begin{theorem}\label{theorem:Lindahl-CES-PRD-substitute}
    Consider the Lindahl equilibrium with agents having CES utilities such that the elasticity parameter $\rho_i \geq 0$ for all $i$. The dynamics converge to the Lindahl equilibrium if agents update their spending according to \eqref{eq:prd:ces:lindahl:pos} if $\rho_i \in (0, 1]$ and $b_{ij}^{t} = B_i \frac{a_{ij}}{\sum_{j'} a_{ij'}}$ if $\rho_i = 0$; and $x_j^t  = \sum_i b_{ij}^t$ for all $j$. Furthermore, the dynamics converge to the Lindahl equilibrium at a rate of $O(1 / T)$. Additionally, if there is no agent with $\rho_i = 1$, then the convergence is $O((\max_i \{\rho_i\})^T)$.
\end{theorem}
\section{Discussions and Future Directions}\label{sec:conclusion-discussion}
In this paper, we propose a unified duality framework for market equilibria in both private goods and public goods markets, with extension to markets with chores. Our framework establishes a fundamental connection between the Fisher market equilibrium and the Lindahl equilibrium: they are equivalent by exchanging the role of allocations and prices in dual markets. This duality connection unifies different research directions and gives many new results on the characterizations, computation, and market dynamics for both the Fisher market equilibrium and the Lindahl equilibrium. 

We believe our framework is general, crystallizes the connection between private goods and public goods markets, and is useful for future works on market equilibrium. We discuss several interesting future directions below.

\paragraph{Capped Setting} In public goods with caps, each public good may have a cap on how much funding (payment) it can receive. The capped setting has applications in stable committee selection~\citep{cheng2020group, munagala2022approximate,gao2025computation, nguyen2025few}. In the capped setting, the Nash welfare maximization solution is not equivalent to the Lindahl equilibrium even with linear utilities. The recent work by \citet{kroer2025computing} show that in the linear case, directly adding the cap constraints to the Shmyrev program yields a Lindahl equilibrium. The efficient computation of a Lindahl equilibrium in the capped setting with more general utilities remains largely open. Is there a general reduction that gives efficient algorithms for the capped setting from algorithms for the uncapped setting?
\paragraph{Markets with Both Public and Private Goods}
In this paper, we analyze markets that consist solely of private goods (Fisher markets) or solely of public goods (Lindahl equilibria). A more realistic model of economics consists of mixed goods, where private and public goods are allocated simultaneously. For instance,a household must allocate its budget to purchase private goods like food and housing, while also funding public goods (parks, schools, etc.) via taxation. Whether such mixed economies can be reduced to an equivalent private-goods-only (Fisher) market is an interesting question. A positive answer would unify the two frameworks and allow algorithmic and structural results for Fisher markets to carry over to mixed-goods settings.

\paragraph{Public Chores} As we have discussed in \Cref{sec:public chores}, the literature on the fair division of public chores is very sparse compared to private chores. We also refer the readers to a recent work~\citep{elkind2025not} for a discussion. We propose the Lindahl equilibrium for public chores as one solution concept and show that it is (weakly) Pareto-optimal. We believe it is an interesting future direction to explore the properties and applications of the Lindahl equilibrium for chores and the public chores setting in general. 

\printbibliography
\appendix
\section{Technical Lemmas on Demand Properties for Gross Substitutes Utilities}
\label{sec:properties-GS}
Recall, given a price vector of goods $\bbp > 0$, and a budget $B_i$, the demand of buyer $i$ is defined as:
\[ \demand_i(\bbp, B_i) \triangleq \argmax_{\bbx_i \in \-R_{\geq 0}^m: \InAngles{\mp, \mx_i} \leq B_i} u_i(\bbx_i). \]
The proofs for the following lemmas are detailed in \cite{cheung2025proportional}. These proofs rely on the assumption that the utility functions satisfy gross substitutes and normal good properties, and are strictly increasing and strictly concave.
\begin{lemma} \label{pec:lem:tech-1}
Given an allocation $\bbx_i$, and two price vectors $\bbp > 0$ and $\bbp' > 0$ such that $\bbx_i = \demand_i(\bbp, B_i) = \demand_i(\bbp', B_i)$,
if $x_{ij} > 0$ for good $j$, then $p_j = p'_j$. 
\end{lemma}
\begin{lemma} \label{pec:lem:main-technique}
     For any price vectors  $\bbp > 0$ and $\bbq > 0$, the following inequality holds:
    \[ \sum_j p_j \demandnb_{ij}(\bbp, B_i) \log \frac{p_j}{q_j}  \leq \sum_j p_j\left[\demandnb_{ij}(\bbq, B_i) - \demandnb_{ij}(\bbp, B_i)\right]. \]
    Equality occurs only if $\demand_{i}(\bbq, B_i) = \demand_{i}(\bbp, B_i)$.
\end{lemma}
\begin{corollary} \label{pec:cor:main-technique}
    For any price vectors $\bbp > 0$ and $\bbq > 0$, and any $1 > \epsilon > 0$, define $\tilde{\bbq}$ as follows: $$\tilde{q}_j = \begin{cases} \epsilon p_j \text{ if } q_j < \epsilon p_j \\ \frac{1}{\epsilon} p_j \text{ if } q_j > \frac{1}{\epsilon} p_j \\ q_j \text{ o.w.} \end{cases}, $$ 
    Then, we have the following inequality:
    \begin{align*} &\sum_j p_j \demandnb_{ij}(\bbp, B_i) \log \frac{p_j}{q_j}  - \sum_j p_j(\demandnb_{ij}(\bbq, B_i) - \demandnb_{ij}(\bbp, B_i))  \\
    &\hspace{1.5in}\leq \sum_j p_j \demandnb_{ij}(\bbp, B_i) \log \frac{p_j}{\tilde{q}_j}  - \sum_j p_j(\demandnb_{ij}(\tilde{\bbq}, B_i) - \demandnb_{ij}(\bbp, B_i)).
    \end{align*}
\end{corollary}
\begin{lemma}\label{pec:lem:tech-2}
    Given an allocation $\bbx_i = \demand_i(\bbp, B_i)$ and $\bbp > 0$, there exists a threshold $\epsilon > 0$ such that $\bbp' \geq \epsilon$ for any $\bbp'$ satisfying $\bbx_i = \demand_i(\bbp', B_i)$.
\end{lemma}
\begin{lemma}\label{pec:lem:tech-3}
    Given an allocation $\bbx_i$ and price vector $\bbp > 0$ and $\bbp' > 0$ such that $\bbx_i = \demand_i(\bbp, B_i) = \demand_i(\bbp', B_i)$. For any $\epsilon < 1$, let $\tilde{\bbp}  = \min\{\bbp/ \epsilon, \bbp'\}$. Then, $\demand_i(\tilde{\bbp}, B_i) = \bbx_i$. 
\end{lemma}
\begin{lemma} \label{pec:lem:converge-allocation-1}
    Let $\{\bbx_i^s > 0\}_s$ be a sequence of allocations and let $\{\bbp^s\}_s$ be the corresponding prices, $\bbx_i^s = \demand_i(\bbp^s, B_i)$. If $\bbx_i^s$ converges to $\bbx_i$ when $s \rightarrow \infty$ and $\bbx_i = \demand_i(\bbp, B_i)$, then $p_j^s x^s_{ij}$ converges to $p_j x_{ij}$. This also implies $p_j^s$ converges to $p_j$ if $x_{ij} > 0$.
\end{lemma}

\section{Useful Facts}
\begin{theorem}[Theorem 1 in \citep{hjertstrand2020homogeneity}]
\label{thm: quasi-concave-2-concave}
    If a function $f: B\subseteq \-R^m_{\ge 0} \rightarrow \-R_{>0}$ is quasi-concave, non-decreasing, and $1$-homogeneous, then $f$ is concave. 
\end{theorem}

\begin{theorem}
\label{thm: quasi-concave-2-concave-new}
    If a function $f: \-R^m_{\ge 0} \rightarrow \-R_{\ge 0}$ is quasi-concave, non-decreasing, and $1$-homogeneous, then $f$ is concave.
\end{theorem}
\begin{proof}
    We note that if $f$ is identically $0$, then the claim holds. Now given any $\mx, \mx' \in \-R^m_{\ge 0}$ and $\alpha \in (0,1)$, we defined $\mx''=\alpha x + (1-\alpha)x'$. If $f(\mx)= 0$, then we know $f(\mx'') \ge f((1-\alpha)\mx') = (1-\alpha) f(\mx') = \alpha f(\mx) + (1-\alpha) f(\mx')$; If $f(\mx') = 0$, the proof is same as the previous case; If both $f(\mx), f(\mx') > 0$, we can apply \Cref{thm: quasi-concave-2-concave}. This completes the proof.
\end{proof}

\end{document}